\documentclass[journal]{IEEEtran}
\usepackage{amsmath,amssymb,bm}
\usepackage{graphicx}
\usepackage{algorithm}
\usepackage{algpseudocode}
\usepackage[colorlinks,citecolor=blue,linkcolor=blue]{hyperref}

\newtheorem{remark}{Remark}
\newtheorem{theorem}{Theorem}
\newtheorem{corollary}{Corollary}

\newtheorem{proposition}{Proposition}

\newcommand{\qed}{\hfill$\blacksquare$}
\newenvironment{proof}{\par\vspace{0.5em}\noindent\textbf{Proof.}\hspace{0.5em}}{\par\qed}

\newcommand{\E}{\mathbb{E}}
\newcommand{\Prob}{\mathbb{P}}
\newcommand{\tr}{\mathrm{tr}}
\newcommand{\diag}{\mathrm{diag}}
\newcommand{\rank}{\mathrm{rank}}
\renewcommand{\Re}{\mathrm{Re}}

\begin{document}

\title{Karhunen-Lo\`{e}ve Expansion for Fluid Antenna Systems: Information-Theoretic Optimal Channel Compression and Outage Analysis}

\author{Tuo~Wu  
\thanks{T. Wu  is with the School of Electronic and Information Engineering, South China University of Technology, Guangzhou 510640, China  (E-mail: $ \rm wtopp0415@163.com$). } 
}
\markboth{}{Wu\MakeLowercase{\textit{et al.}}: KL Expansion for FAS}

\maketitle

\begin{abstract}
Fluid antenna systems (FAS) achieve spatial diversity by dynamically switching among $N$ densely packed ports, but the resulting spatially correlated Rayleigh channels render exact outage analysis intractable. Existing block-correlation models (BCM) impose structural approximations on the channel covariance matrix that can introduce optimistic performance bias. This paper proposes a principled Karhunen-Lo\`{e}ve (KL) expansion framework that decomposes the $N$-dimensional correlated FAS channel into independent eigenmodes and performs a controlled rank-$K$ truncation, reducing the outage analysis to a $K$-dimensional integration with $K \ll N$. Closed-form outage expressions are derived for the rank-1 and rank-2 cases, and a general Gauss-Hermite quadrature formula is provided for arbitrary $K$. On the theoretical front, it is proved via Anderson's inequality that the KL approximation \emph{always} overestimates the outage probability, providing a conservative guarantee essential for secure system design. Leveraging the Slepian--Landau--Pollak concentration theorem, it is established that only $K^* = 2\lceil W \rceil + 1$ eigenmodes are needed regardless of $N$, where $W$ is the normalized aperture. It is further shown that the KL truncation achieves the Gaussian rate-distortion bound, certifying it as the information-theoretically optimal channel compression. Extensive numerical results confirm that (i) theoretical predictions match Monte Carlo simulations, (ii) the entropy fraction converges faster than the power fraction, (iii) the KL framework uniformly outperforms BCM in approximation accuracy while avoiding the optimistic bias inherent in block-diagonal models, and (iv) the effective degrees of freedom scale with the aperture rather than the number of ports.
\end{abstract}


\section{Introduction}

\IEEEPARstart{F}{luid} antenna systems (FAS), also referred to as movable antennas~\cite{Zhu-Wong-2024,LZhu25}, have emerged as a transformative technology for sixth-generation (6G) wireless networks~\cite{TWu20243,New24}. Unlike conventional fixed-position antenna arrays, FAS enables a single radio-frequency (RF) chain to dynamically switch among $N$ densely packed ports within a compact linear aperture of $W\lambda$, achieving spatial diversity gains that rival multi-antenna systems at a fraction of the hardware cost~\cite{FAS21,FAS20}. The core principle is that the propagation environment offers rich spatial variation even within a small aperture: by selecting the port with the strongest instantaneous channel gain, FAS exploits this spatial richness without requiring multiple RF chains or complex beamforming hardware. The evolution of FAS from a conceptual framework to a mature research field has been documented in a series of overview articles~\cite{MFAS23,Shojaeifard}, which identify open problems spanning channel modeling, port selection algorithms, and hardware co-design. Practical hardware implementations include pixel-based reconfigurable antennas~\cite{Rodrigo14,Zhang25}, liquid-metal radiators~\cite{Huang21}, and programmable meta-fluid antennas~\cite{BLiu25}. This architectural advantage has catalyzed a rapidly growing body of research on FAS-aided communications, spanning outage analysis~\cite{FAS20,Chai22,Ghadi-2023,New-twc2023}, multiple access~\cite{FAMS,FAMS23,NWaqar23}, channel estimation~\cite{HXu23}, port selection under maximum ratio combining~\cite{XLai23}, physical layer security~\cite{Security1,Security2,Security3,Security5,TuoW}, cognitive radio~\cite{YaoJ241}, and integrated sensing and communication (ISAC)~\cite{CWang24,LZhouWCL24}. Among these research directions, outage probability analysis occupies a central role: it directly characterizes the reliability of the port-selection diversity gain and serves as the foundation for system design, port density optimization, and aperture sizing.

Despite this rapid progress, a fundamental bottleneck persists in FAS outage analysis. Since adjacent ports share common scattering environments, the channel vector $\mathbf{g} = [g_1, \ldots, g_N]^T$ follows a correlated complex Gaussian distribution whose covariance matrix $\mathbf{R}$ exhibits a Toeplitz structure under the widely adopted Jakes model~\cite{FAS22}. Computing the outage probability therefore requires the joint CDF of $N$ correlated Rayleigh random variables---an $N$-dimensional integral that is analytically intractable for general $N$ and whose computational cost grows exponentially with $N$. This intractability is not merely a technical inconvenience: it fundamentally limits the ability to perform closed-form system optimization, derive diversity-multiplexing tradeoffs, or obtain design insights that scale with the system parameters.

To circumvent this difficulty, block-correlation models (BCM) have been proposed~\cite{BC24,LaiX242}, which partition the $N$ ports into $D$ independent blocks with a common intra-block correlation coefficient, reducing the $N$-dimensional integral to a product of $D$ one-dimensional integrals. The variable BCM (VBCM) extends this by allowing block-specific correlation coefficients~\cite{LaiX}, partially reducing the approximation error. While these models yield tractable closed-form expressions, they introduce a fundamental structural distortion: by imposing a block-diagonal constraint on $\mathbf{R}$, they eliminate all inter-block correlations present in the true Toeplitz matrix. This removal artificially inflates the effective diversity order---since the true channel is more correlated than the block-diagonal model suggests---and produces a systematic \emph{optimistic bias} in outage estimation. The BCM/VBCM consistently underestimate the true outage probability, which is particularly dangerous in security-critical applications~\cite{TuoW,Security1}: an underestimated outage implies an overestimated secrecy rate, potentially causing the system to transmit at a rate that is not actually secure. Moreover, as the number of ports $N$ grows, the BCM/VBCM must increase the number of blocks $D$ proportionally, making them inherently non-scalable.

The key observation motivating this paper is that the Jakes correlation matrix $\mathbf{R}$ possesses a \emph{rapidly decaying eigenvalue spectrum} that is fundamentally different from its block-diagonal approximation. Specifically, by the Slepian--Landau--Pollak concentration theorem~\cite{Slepian61,Landau62}, the eigenvalues of $\mathbf{R}$ exhibit a sharp phase transition at index $K^* = 2\lceil W \rceil + 1$: the first $K^*$ eigenvalues are $\Theta(N/K^*)$ while the remaining $N - K^*$ eigenvalues decay super-exponentially to zero. Crucially, $K^*$ depends only on the normalized aperture $W$, not on the number of ports $N$. This spectral concentration property suggests a principled alternative to BCM/VBCM: decompose the channel vector $\mathbf{g}$ via its Karhunen-Lo\`{e}ve (KL) expansion $\mathbf{g} = \sqrt{\eta}\,\mathbf{U}\boldsymbol{\Lambda}^{1/2}\mathbf{z}$ into $N$ statistically independent eigenmodes $\{z_k\}$, retain only the $K$ dominant ones, and reduce the $N$-dimensional correlated outage problem to a $K$-dimensional independent one. Unlike BCM/VBCM, this approach imposes no structural constraint on $\mathbf{R}$, works directly with the true eigenbasis, and admits a rigorous truncation error bound. The resulting $K$-dimensional integral can be evaluated efficiently via Gauss-Hermite quadrature, with closed-form expressions available for the rank-1 and rank-2 special cases.

Realizing this approach with full analytical rigor, however, requires overcoming three non-trivial challenges. \textit{First}, while the KL truncation reduces the channel power, it is not immediately clear whether this reduction leads to a conservative or optimistic outage estimate. Establishing the direction of the bias requires invoking Anderson's inequality~\cite{Anderson55} for symmetric convex sets---a result from geometric probability that, when applied to the outage event (a product of $N$ disks in $\mathbb{C}^N$), proves that the KL-truncated outage is always an upper bound on the true outage, with monotone convergence as $K$ increases. \textit{Second}, establishing that $K^* = \mathcal{O}(2W+1)$ suffices for near-exact outage prediction requires a careful asymptotic analysis connecting the Toeplitz eigenvalue distribution to the Szeg\H{o}--Slepian spectral concentration theory, and making this connection precise for finite $N$. \textit{Third}, certifying the KL truncation as information-theoretically optimal---in the sense that no other rank-$K$ approximation of $\mathbf{R}$ can achieve a lower channel compression distortion at the same rate---requires connecting the KL truncation to the Gaussian rate-distortion function via reverse water-filling, and showing that the block-diagonal constraint of BCM/VBCM prevents them from achieving this fundamental bound.

The main contributions of this paper are summarized as follows:
\begin{itemize}
	\item \textbf{KL expansion framework for FAS outage analysis:} This paper derives the outage probability of FAS under the KL-truncated channel model, reducing the $N$-dimensional correlated problem to a $K$-dimensional independent one. Closed-form expressions are obtained for the rank-1 and rank-2 cases, and a general Gauss-Hermite quadrature formula is provided for arbitrary $K$.
	
	\item \textbf{Conservative outage guarantee:} By invoking Anderson's inequality on symmetric convex sets, this paper proves that the KL-truncated outage probability is always an upper bound on the true outage, with monotone convergence as $K$ increases. This conservative property is critical for secure system design, where optimistic estimates can lead to security breaches.
	
	\item \textbf{Effective degrees of freedom:} Leveraging the Slepian--Landau--Pollak concentration theorem, this paper establishes that the effective number of spatial degrees of freedom is $K^* = 2\lceil W \rceil + 1$, depending only on the normalized aperture $W$ and independent of the port count $N$. This provides a theoretical foundation for the scalability of the KL approach.
	
	\item \textbf{Rate-distortion optimality:} This paper proves that the KL truncation achieves the Gaussian rate-distortion bound via reverse water-filling, confirming it as the information-theoretically optimal channel compression. In contrast, the BCM/VBCM, constrained to a block-diagonal structure, operates strictly above this bound.
	
	\item \textbf{Comprehensive numerical verification:} Extensive Monte Carlo simulations are provided to validate all theoretical predictions, including the conservative outage guarantee, the eigenvalue phase transition, the entropy-power duality, the ergodic capacity lower bound, and the rate-distortion optimality. A systematic comparison with BCM, VBCM, and i.i.d.\ models is also provided.
\end{itemize}

The remainder of this paper is organized as follows. Section~II presents the system model. Section~III develops the KL expansion and low-rank truncation. Section~IV derives the outage probability under the KL framework. Section~V establishes the information-theoretic guarantees. Section~VI provides numerical results, and Section~VII concludes the paper.

\section{System Model}

Before developing the KL expansion framework, this section establishes the system model and formally defines the outage probability that serves as the central performance metric throughout the paper. Table~\ref{tab:notation} summarizes the key notation used throughout the paper.

\begin{table}[t]
	\centering
	\caption{Summary of Key Notation}
	\label{tab:notation}
	\footnotesize
	\renewcommand{\arraystretch}{1.15}
	\begin{tabular}{cl}
		\hline
		\textbf{Symbol} & \textbf{Definition} \\
		\hline
		$N$ & Number of FAS ports \\
		$W$ & Normalized aperture ($W\lambda$ = physical aperture) \\
		$\lambda$ & Carrier wavelength \\
		$P$ & Transmit power \\
		$\bar{\gamma}$ & Average received SNR, $\bar{\gamma} = P\eta d^{-a}/\sigma^2$ \\
		$\gamma_{\mathrm{th}}$ & SNR outage threshold \\
		$\eta$ & Mean channel power per port \\
		$g_n$ & Complex channel coefficient at port $n$ \\
		$\mathbf{g}$ & Channel vector $[g_1,\ldots,g_N]^T \in \mathbb{C}^N$ \\
		$\mathbf{R}$ & $N\times N$ Toeplitz spatial correlation matrix \\
		$[\mathbf{R}]_{m,n}$ & Jakes correlation $J_0(2\pi W|m-n|/(N-1))$ \\
		$\mathbf{U},\boldsymbol{\Lambda}$ & Eigenvector matrix and diagonal eigenvalue matrix of $\mathbf{R}$ \\
		$\lambda_k$ & $k$-th eigenvalue of $\mathbf{R}$, ordered $\lambda_1\geq\cdots\geq\lambda_N\geq 0$ \\
		$\mathbf{u}_k$ & $k$-th eigenvector (column of $\mathbf{U}$) \\
		$z_k$ & $k$-th KL coefficient, $z_k\sim\mathcal{CN}(0,1)$, i.i.d. \\
		$K$ & Number of retained KL modes (truncation order) \\
		$K^*$ & Effective degrees of freedom, $K^*=2\lceil W\rceil+1$ \\
		$\tilde{\mathbf{g}}^{(K)}$ & Rank-$K$ KL-truncated channel vector \\
		$\varepsilon_K$ & Fractional power loss due to truncation \\
		$\tilde{F}_{\max}^{(K)}(x)$ & CDF of normalized max gain under $K$-mode truncation \\
		$P_{\mathrm{out}}$ & True outage probability \\
		$\tilde{P}_{\mathrm{out}}^{(K)}$ & KL-approximated outage probability \\
		$\bar{C}$ & Ergodic capacity under exact channel \\
		$\bar{C}_{\mathrm{KL}}^{(K)}$ & Ergodic capacity under $K$-mode KL truncation \\
		$c_1$ & Peak magnitude of dominant eigenvector, $\max_n|u_{n,1}|^2$ \\
		$E_1(\cdot)$ & Exponential integral, $E_1(x)=\int_x^\infty t^{-1}e^{-t}dt$ \\
		$J_0(\cdot)$ & Zeroth-order Bessel function of the first kind \\
		$\mathbf{A}\preceq\mathbf{B}$ & Loewner (PSD) ordering: $\mathbf{B}-\mathbf{A}\succeq\mathbf{0}$ \\
		$\|\cdot\|_F$ & Frobenius norm \\
		\hline
	\end{tabular}
\end{table}

The system under consideration is a point-to-point wireless communication system in which a single-antenna base station (BS) transmits to a user equipped with a one-dimensional fluid antenna system (FAS) capable of switching among $N$ discrete ports within a linear space of $W\lambda$, where $W$ is the normalized antenna aperture and $\lambda$ is the carrier wavelength. The system operates under quasi-static block Rayleigh fading.

\subsection{Signal Model}

The BS transmits the signal $s \sim \mathcal{CN}(0,1)$ with transmit power $P$. The received signal at the $n$-th port of the FAS is
\begin{align}\label{eq:signal}
	y_n = \sqrt{P}\, d^{-a/2}\, g_n\, s + n_n, \quad n = 1, \ldots, N,
\end{align}
where $g_n \sim \mathcal{CN}(0, \eta)$ is the complex channel coefficient from the BS to the $n$-th FAS port, $d$ is the BS--user distance, $a$ is the path-loss exponent, and $n_n \sim \mathcal{CN}(0, \sigma^2)$ is the additive white Gaussian noise (AWGN).

The FAS selects the optimal port that maximizes the instantaneous channel gain:
\begin{align}\label{eq:port_select}
	n^* = \arg\max_{n \in \{1,\ldots,N\}} |g_n|^2.
\end{align}
The resulting instantaneous signal-to-noise ratio (SNR) is
\begin{align}\label{eq:snr}
	\gamma = \bar{\gamma} \cdot \max_{1 \leq n \leq N} \frac{|g_n|^2}{\eta},
\end{align}
where $\bar{\gamma} \triangleq P d^{-a} \eta / \sigma^2$ denotes the average SNR.

\subsection{FAS Channel Correlation Model}

The channel vector $\mathbf{g} = [g_1, \ldots, g_N]^T \sim \mathcal{CN}(\mathbf{0}, \mathbf{\Sigma})$ is spatially correlated due to the close proximity of adjacent FAS ports. The covariance matrix is $\mathbf{\Sigma} = \eta\, \mathbf{R}$, where $\mathbf{R}$ is the normalized correlation matrix. Under the Jakes model, the $(k,l)$-th entry of $\mathbf{R}$ is
\begin{align}\label{eq:jakes}
	[\mathbf{R}]_{k,l} = J_0\!\left(\frac{2\pi |k-l| W}{N-1}\right),
\end{align}
where $J_0(\cdot)$ is the zeroth-order Bessel function of the first kind. The matrix $\mathbf{R}$ is Hermitian positive semi-definite with a symmetric Toeplitz structure. Its diagonal entries satisfy $[\mathbf{R}]_{n,n} = 1$ for all $n$, so $\tr(\mathbf{R}) = N$.

\subsection{Outage Probability}

The outage probability is defined as the probability that the instantaneous SNR falls below a target threshold $\gamma_{\mathrm{th}}$:
\begin{align}\label{eq:pout_def}
	P_{\mathrm{out}} = \Prob\!\left(\gamma \leq \gamma_{\mathrm{th}}\right) = F_{\max}\!\left(\frac{\gamma_{\mathrm{th}}}{\bar{\gamma}}\right),
\end{align}
where $F_{\max}(x) \triangleq \Prob\!\left(\max_{n} |g_n|^2/\eta \leq x\right)$ is the CDF of the normalized maximum channel gain.

\begin{remark}[Analytical Challenge]
	Computing $F_{\max}(x)$ requires the joint CDF of $N$ correlated Rayleigh random variables, which involves the $N$-dimensional integral
	\begin{align}
		F_{\max}(x) = \int \cdots \int_{\mathcal{D}(x)} f_{\mathbf{g}}(\mathbf{g})\, d\mathbf{g},
	\end{align}
	where $\mathcal{D}(x) = \{\mathbf{g} \in \mathbb{C}^N : |g_n|^2/\eta \leq x,\ \forall n\}$ and $f_{\mathbf{g}}$ is the multivariate complex Gaussian density. Due to the non-diagonal Toeplitz structure of $\mathbf{R}$, this integral does not admit a product form and is analytically intractable for general $N$. This motivates the KL expansion approach developed in the next section.
\end{remark}

\section{Karhunen-Lo\`{e}ve Expansion of the FAS Channel}

With the system model and the intractability of the exact outage integral established, this section develops the KL expansion of the FAS channel. The key idea is to diagonalize the correlated channel via its eigenbasis, enabling a controlled low-rank approximation that preserves the true spectral structure of $\mathbf{R}$ without any block-diagonal constraint.

\subsection{Eigendecomposition of the Correlation Matrix}

The eigendecomposition of the correlation matrix is performed as follows:
\begin{align}\label{eq:eigen}
	\mathbf{R} = \mathbf{U} \mathbf{\Lambda} \mathbf{U}^H,
\end{align}
where $\mathbf{U} = [\mathbf{u}_1, \ldots, \mathbf{u}_N] \in \mathbb{C}^{N \times N}$ is the unitary eigenvector matrix and $\mathbf{\Lambda} = \diag(\lambda_1, \ldots, \lambda_N)$ contains the eigenvalues ordered as $\lambda_1 \geq \lambda_2 \geq \cdots \geq \lambda_N \geq 0$. Since $\mathbf{R}$ is a correlation matrix with unit diagonal, the eigenvalues satisfy
\begin{align}\label{eq:trace_constraint}
	\sum_{k=1}^{N} \lambda_k = \tr(\mathbf{R}) = N.
\end{align}
Let $u_{n,k} \triangleq [\mathbf{U}]_{n,k}$ denote the $(n,k)$-th entry of $\mathbf{U}$. The unitarity of $\mathbf{U}$ guarantees that for each port $n$:
\begin{align}\label{eq:unitary_row}
	\sum_{k=1}^{N} |u_{n,k}|^2 = 1, \quad \forall\, n = 1, \ldots, N.
\end{align}

\subsection{KL Representation of the Channel}

Using \eqref{eq:eigen}, the channel vector admits the Karhunen-Lo\`{e}ve (KL) expansion:
\begin{align}\label{eq:kl_full}
	\mathbf{g} = \sqrt{\eta}\, \mathbf{U} \mathbf{\Lambda}^{1/2} \mathbf{z},
\end{align}
where $\mathbf{z} = [z_1, \ldots, z_N]^T \sim \mathcal{CN}(\mathbf{0}, \mathbf{I}_N)$ is a vector of independent and identically distributed (i.i.d.) standard complex Gaussian random variables. The channel coefficient at port $n$ is thus expressed as
\begin{align}\label{eq:kl_port}
	g_n = \sqrt{\eta} \sum_{k=1}^{N} \sqrt{\lambda_k}\, u_{n,k}\, z_k.
\end{align}

\begin{remark}[Physical Interpretation of Eigenmodes]
	Each eigenmode $z_k$ represents an independent spatial degree of freedom of the FAS channel. The eigenvalue $\lambda_k$ quantifies the power captured by the $k$-th mode. In the presence of strong spatial correlation (small $W$ or large $N/W$), the eigenvalue spectrum decays rapidly: a few dominant modes capture most of the channel energy, while the remaining modes contribute negligible power. This spectral concentration is the key property that enables low-rank truncation.
\end{remark}

\subsection{Low-Rank Truncation}

The channel is approximated by retaining only the $K$ dominant eigenmodes ($K \ll N$):
\begin{align}\label{eq:kl_trunc}
	\tilde{g}_n = \sqrt{\eta} \sum_{k=1}^{K} \sqrt{\lambda_k}\, u_{n,k}\, z_k, \quad n = 1, \ldots, N.
\end{align}
The truncated channel vector is $\tilde{\mathbf{g}} = \sqrt{\eta}\, \mathbf{U}_K \mathbf{\Lambda}_K^{1/2} \mathbf{z}_K$, where $\mathbf{U}_K \in \mathbb{C}^{N \times K}$ contains the first $K$ columns of $\mathbf{U}$, $\mathbf{\Lambda}_K = \diag(\lambda_1, \ldots, \lambda_K)$, and $\mathbf{z}_K = [z_1, \ldots, z_K]^T$.

The truncation error at port $n$ is
\begin{align}\label{eq:trunc_error_port}
	\Delta_n \triangleq g_n - \tilde{g}_n = \sqrt{\eta} \sum_{k=K+1}^{N} \sqrt{\lambda_k}\, u_{n,k}\, z_k,
\end{align}
with variance
\begin{align}\label{eq:trunc_var}
	\E[|\Delta_n|^2] = \eta \sum_{k=K+1}^{N} \lambda_k |u_{n,k}|^2 \leq \eta \lambda_{K+1},
\end{align}
where the inequality follows from $\sum_{k=K+1}^{N} |u_{n,k}|^2 \leq 1$ (by \eqref{eq:unitary_row}) and $\lambda_k \leq \lambda_{K+1}$ for $k \geq K+1$.

\begin{proposition}[Truncation Error Bound]\label{prop:trunc_error}
	Let $F_{\max}(x)$ and $\tilde{F}_{\max}^{(K)}(x)$ denote the CDFs of $\max_n |g_n|^2/\eta$ under the full and $K$-truncated KL expansions, respectively. The fraction of total channel power lost due to truncation is bounded as
	\begin{align}\label{eq:power_loss}
		\varepsilon_K \triangleq \frac{\E[\|\mathbf{g} - \tilde{\mathbf{g}}\|^2]}{\E[\|\mathbf{g}\|^2]} = \frac{\sum_{k=K+1}^{N}\lambda_k}{\sum_{k=1}^{N}\lambda_k} = 1 - \frac{\sum_{k=1}^{K}\lambda_k}{N}.
	\end{align}
\end{proposition}

\begin{proof}
	Let $\mathbf{U}_{r} = [\mathbf{u}_{K+1}, \ldots, \mathbf{u}_N] \in \mathbb{C}^{N \times (N-K)}$ and $\boldsymbol{\Lambda}_{r} = \diag(\lambda_{K+1}, \ldots, \lambda_N)$ denote the matrices formed by the discarded eigenvectors and eigenvalues, respectively. The truncation residual is
	\begin{align}
		\mathbf{g} - \tilde{\mathbf{g}} = \sqrt{\eta}\,\mathbf{U}\boldsymbol{\Lambda}^{1/2}\mathbf{z} - \sqrt{\eta}\,\mathbf{U}_K\boldsymbol{\Lambda}_K^{1/2}\mathbf{z}_K = \sqrt{\eta}\, \mathbf{U}_{r} \boldsymbol{\Lambda}_{r}^{1/2} \mathbf{z}_{r},
	\end{align}
	where $\mathbf{z}_r = [z_{K+1}, \ldots, z_N]^T \sim \mathcal{CN}(\mathbf{0}, \mathbf{I}_{N-K})$ is independent of $\mathbf{z}_K$. The expected squared norm of the residual is
	\begin{align}
		\E[\|\mathbf{g} - \tilde{\mathbf{g}}\|^2] &= \eta\, \E\!\left[\mathbf{z}_r^H \boldsymbol{\Lambda}_r^{1/2} \mathbf{U}_r^H \mathbf{U}_r \boldsymbol{\Lambda}_r^{1/2} \mathbf{z}_r\right] \nonumber\\
		&= \eta\, \E\!\left[\mathbf{z}_r^H \boldsymbol{\Lambda}_r \mathbf{z}_r\right] \nonumber\\
		&= \eta\, \tr\!\left(\boldsymbol{\Lambda}_r \E[\mathbf{z}_r \mathbf{z}_r^H]\right) \nonumber\\
		&= \eta\, \tr(\boldsymbol{\Lambda}_r) = \eta \sum_{k=K+1}^{N} \lambda_k,
	\end{align}
	where the second equality uses $\mathbf{U}_r^H \mathbf{U}_r = \mathbf{I}_{N-K}$ (since $\mathbf{U}$ is unitary), and the third uses $\E[\mathbf{z}_r \mathbf{z}_r^H] = \mathbf{I}_{N-K}$. The total channel power is
	\begin{align}
		\E[\|\mathbf{g}\|^2] = \eta\, \tr(\mathbf{R}) = \eta \sum_{k=1}^{N} \lambda_k = \eta N,
	\end{align}
	where the last equality uses $\tr(\mathbf{R}) = N$ (unit diagonal). Dividing yields \eqref{eq:power_loss}.
\end{proof}

\begin{remark}[Choosing $K$]
	Given a target approximation accuracy $\varepsilon_0 > 0$, the minimum number of retained modes is
	\begin{align}\label{eq:K_choice}
		K^* = \min\left\{K : \sum_{k=1}^{K} \lambda_k \geq (1-\varepsilon_0)\, N\right\}.
	\end{align}
	For FAS channels under the Jakes model with moderate aperture $W$, the eigenvalue spectrum of $\mathbf{R}$ decays rapidly, and typically $K^* = \mathcal{O}(\lceil 2W \rceil + 1)$ suffices to capture over 99\% of the total power, regardless of $N$. Crucially, $K^*$ depends on the aperture $W$ but not on $N$, making this approach scalable to large port numbers.
\end{remark}

\section{Outage Probability via KL Expansion}

Having established the KL expansion and the truncation error bound, this section derives the outage probability under the rank-$K$ truncated channel model. The derivation proceeds from the general $K$-dimensional integral formulation, through Gauss-Hermite quadrature evaluation, to closed-form expressions for the rank-1 and rank-2 special cases, and finally to the ergodic rate.

\subsection{General \texorpdfstring{$K$}{K}-Dimensional Integral Formulation}

Under the $K$-truncated KL expansion, the CDF of the normalized maximum channel gain is
\begin{align}\label{eq:cdf_kl_general}
	\tilde{F}_{\max}^{(K)}(x) = \E_{\mathbf{z}_K}\!\left[\prod_{n=1}^{N} \mathbf{1}\!\left(\left|\sum_{k=1}^{K} \sqrt{\lambda_k}\, u_{n,k}\, z_k \right|^2 \leq x\right)\right],
\end{align}
where $\mathbf{1}(\cdot)$ is the indicator function and $z_k \sim \mathcal{CN}(0,1)$ are i.i.d. This is a $K$-dimensional expectation over independent complex Gaussian variables, which reduces the original $N$-dimensional correlated problem to a $K$-dimensional independent one.

The corresponding outage probability is
\begin{align}\label{eq:pout_kl}
	\tilde{P}_{\mathrm{out}}^{(K)} = \tilde{F}_{\max}^{(K)}\!\left(\frac{\gamma_{\mathrm{th}}}{\bar{\gamma}}\right).
\end{align}

\subsection{Gauss-Hermite Quadrature Evaluation}

To evaluate \eqref{eq:cdf_kl_general} numerically, each complex mode is written as $z_k = z_k^R + j\, z_k^I$ where $z_k^R, z_k^I \sim \mathcal{N}(0, 1/2)$ are independent. The probability density kernel $e^{-(z_k^R)^2 - (z_k^I)^2}$ directly matches the Gauss-Hermite weight function. After substitution, \eqref{eq:cdf_kl_general} becomes a $2K$-dimensional real integral:
\begin{align}\label{eq:cdf_real}
	\tilde{F}_{\max}^{(K)}(x) = \frac{1}{\pi^K} \int_{\mathbb{R}^{2K}} \Psi(\mathbf{t};\, x)\, \prod_{k=1}^{K} e^{-(t_k^R)^2 - (t_k^I)^2}\, d\mathbf{t}^R\, d\mathbf{t}^I,
\end{align}
where $\mathbf{t} = (t_1^R, t_1^I, \ldots, t_K^R, t_K^I)$ and
\begin{align}\label{eq:psi_def}
	\Psi(\mathbf{t};\, x) = \prod_{n=1}^{N} \mathbf{1}\!\left(\left|\sum_{k=1}^{K} \sqrt{\lambda_k}\, u_{n,k}\, (t_k^R + j\, t_k^I) \right|^2 \leq x\right).
\end{align}

Applying the $2K$-dimensional Gauss-Hermite quadrature with $Q$ nodes per dimension:
\begin{align}\label{eq:gauss_hermite}
	\tilde{F}_{\max}^{(K)}(x) \approx \frac{1}{\pi^K} \sum_{i_1=1}^{Q} \cdots \sum_{i_{2K}=1}^{Q} \left(\prod_{l=1}^{2K} w_{i_l}\right) \Psi(\mathbf{t}_{\mathbf{i}};\, x),
\end{align}
where $\{t_j, w_j\}_{j=1}^{Q}$ are the standard Gauss-Hermite quadrature nodes and weights, and $\mathbf{t}_{\mathbf{i}}$ denotes the evaluation point. The computational complexity is $\mathcal{O}(Q^{2K} \cdot N)$.

\begin{remark}[Complexity Comparison]
	The original problem requires the $N$-dimensional joint CDF evaluation, which scales as $\mathcal{O}(Q^{2N})$ -- infeasible for even moderate $N$. The KL-based approach reduces this to $\mathcal{O}(Q^{2K})$, where $K \ll N$. For typical FAS parameters ($W \leq 5$), $K \leq 11$ and with $Q = 10$ the computation is efficient. In contrast, the VBCM requires computing products of block CDFs involving Marcum Q-functions with numerical integration per block.
\end{remark}

\subsection{Special Case: Rank-1 Approximation (\texorpdfstring{$K = 1$}{K=1})}

When only the dominant eigenmode is retained ($K=1$), the channel at port $n$ reduces to
\begin{align}\label{eq:rank1_channel}
	\tilde{g}_n^{(1)} = \sqrt{\eta\lambda_1}\, u_{n,1}\, z_1,
\end{align}
where $z_1 \sim \mathcal{CN}(0,1)$. The squared magnitude at port $n$ is
\begin{align}
	|\tilde{g}_n^{(1)}|^2 = \eta\lambda_1 |u_{n,1}|^2 |z_1|^2.
\end{align}
Since $u_{n,1}$ and $\lambda_1$ are deterministic, all $N$ ports are perfectly correlated through the single random variable $z_1$. The maximum over all ports is therefore
\begin{align}\label{eq:max_rank1}
	\max_{1\leq n\leq N} |\tilde{g}_n^{(1)}|^2 = \eta\lambda_1 \left(\max_{1\leq n\leq N} |u_{n,1}|^2\right) |z_1|^2 = \eta\lambda_1 c_1 |z_1|^2,
\end{align}
where $c_1 \triangleq \max_{1 \leq n \leq N} |u_{n,1}|^2 \in (0,1]$ is a deterministic constant determined by the peak magnitude of the dominant eigenvector. Since $z_1 \sim \mathcal{CN}(0,1)$, it follows that $|z_1|^2 \sim \mathrm{Exp}(1)$, i.e., $\Prob(|z_1|^2 \leq t) = 1 - e^{-t}$ for $t \geq 0$. The CDF of the normalized maximum channel gain is therefore
\begin{align}
	\tilde{F}_{\max}^{(1)}(x) &= \Prob\!\left(\frac{\max_n |\tilde{g}_n^{(1)}|^2}{\eta} \leq x\right) = \Prob\!\left(\lambda_1 c_1 |z_1|^2 \leq x\right) \nonumber\\
	&= \Prob\!\left(|z_1|^2 \leq \frac{x}{\lambda_1 c_1}\right) = 1 - e^{-x/(\lambda_1 c_1)}.\label{eq:cdf_rank1}
\end{align}
Substituting $x = \gamma_{\mathrm{th}}/\bar{\gamma}$, the outage probability is
\begin{align}\label{eq:pout_rank1}
	\tilde{P}_{\mathrm{out}}^{(1)} = \tilde{F}_{\max}^{(1)}\!\left(\frac{\gamma_{\mathrm{th}}}{\bar{\gamma}}\right) = 1 - \exp\!\left(-\frac{\gamma_{\mathrm{th}}}{\bar{\gamma}\lambda_1 c_1}\right).
\end{align}

\begin{remark}[Interpretation of Rank-1]
	The rank-1 approximation yields a simple Rayleigh fading model with an effective channel gain $\lambda_1 c_1$. No spatial diversity is captured, since all ports are perfectly correlated under a single mode. The product $\lambda_1 c_1$ quantifies the \emph{beamforming gain}: $\lambda_1$ is the power concentration in the dominant mode, and $c_1$ is the peak spatial focusing of the dominant eigenvector. This approximation is tight only when $\lambda_1 \gg \lambda_2$, i.e., the channel is nearly rank-deficient.
\end{remark}

\subsection{Special Case: Rank-2 Approximation (\texorpdfstring{$K = 2$}{K=2})}

Retaining two eigenmodes ($K=2$), the channel at port $n$ becomes
\begin{align}\label{eq:rank2_channel}
	\tilde{g}_n^{(2)} = \sqrt{\eta}\!\left(\sqrt{\lambda_1}\, u_{n,1}\, z_1 + \sqrt{\lambda_2}\, u_{n,2}\, z_2\right),
\end{align}
where $z_1, z_2 \sim \mathcal{CN}(0,1)$ are independent. Expanding the squared amplitude:
\begin{align}\label{eq:rank2_power}
	\frac{|\tilde{g}_n^{(2)}|^2}{\eta} &= \left|\sqrt{\lambda_1}\, u_{n,1}\, z_1 + \sqrt{\lambda_2}\, u_{n,2}\, z_2\right|^2 \nonumber\\
	&= \lambda_1 |u_{n,1}|^2 |z_1|^2 + \lambda_2 |u_{n,2}|^2 |z_2|^2 \nonumber\\
	&\quad + 2\sqrt{\lambda_1 \lambda_2}\, \Re\!\left(u_{n,1} u_{n,2}^* z_1 z_2^*\right).
\end{align}
The cross-term $\Re(u_{n,1} u_{n,2}^* z_1 z_2^*)$ couples $z_1$ and $z_2$, so the joint distribution of $(\tilde{g}_1^{(2)}, \ldots, \tilde{g}_N^{(2)})$ given $z_1$ and $z_2$ does not factorize, preventing a direct product-form CDF.

To make progress, the derivation conditions on $z_1$ and treats $z_2$ as the remaining random variable. Define
\begin{align}\label{eq:a_n_def}
	a_n(z_1) \triangleq \sqrt{\lambda_1}\, u_{n,1}\, z_1, \quad b_n \triangleq \sqrt{\lambda_2}\, u_{n,2},
\end{align}
so that $\tilde{g}_n^{(2)} / \sqrt{\eta} = a_n(z_1) + b_n z_2$. The outage constraint at port $n$ is $|\tilde{g}_n^{(2)}|^2/\eta \leq x$, which becomes
\begin{align}
	|a_n(z_1) + b_n z_2|^2 \leq x.
\end{align}
Dividing both sides by $|b_n|^2 = \lambda_2 |u_{n,2}|^2$ (assuming $u_{n,2} \neq 0$) and completing the square:
\begin{align}
	\left|\frac{a_n(z_1)}{b_n} + z_2\right|^2 \leq \frac{x}{|b_n|^2},
\end{align}
which is equivalently written as
\begin{align}\label{eq:disk_constraint}
	\left|z_2 - \left(-\frac{a_n(z_1)}{b_n}\right)\right|^2 \leq r_n^2, \quad r_n \triangleq \sqrt{\frac{x}{\lambda_2 |u_{n,2}|^2}}.
\end{align}
This defines a disk $\mathcal{B}_n(z_1)$ in the complex $z_2$-plane centered at $c_n(z_1) = -a_n(z_1)/b_n$ with radius $r_n$. Note that $r_n$ is independent of $z_1$ (it depends only on the fixed threshold $x$ and the eigenvector entry $u_{n,2}$), while the center $c_n(z_1)$ shifts linearly with $z_1$.

The outage event $\{\max_n |\tilde{g}_n^{(2)}|^2/\eta \leq x\}$ requires all $N$ port constraints to hold simultaneously, i.e., $z_2 \in \bigcap_{n=1}^N \mathcal{B}_n(z_1)$. The conditional CDF given $z_1$ is therefore
\begin{align}\label{eq:cdf_rank2_cond}
	\tilde{F}_{\max}^{(2)}(x \mid z_1) = \Prob\!\left(z_2 \in \bigcap_{n=1}^{N} \mathcal{B}_n(z_1)\,\Big|\, z_1\right),
\end{align}
where $z_2 \sim \mathcal{CN}(0,1)$ independently of $z_1$. The unconditional CDF is obtained by averaging over $z_1$:
\begin{align}\label{eq:cdf_rank2}
	\tilde{F}_{\max}^{(2)}(x) = \E_{z_1}\!\left[\tilde{F}_{\max}^{(2)}(x \mid z_1)\right] = \int_{\mathbb{C}} \tilde{F}_{\max}^{(2)}(x \mid z_1)\, \frac{e^{-|z_1|^2}}{\pi}\, d^2 z_1.
\end{align}
The inner probability \eqref{eq:cdf_rank2_cond} is the probability that a $\mathcal{CN}(0,1)$ variable falls within the intersection of $N$ disks whose centers depend on $z_1$. Writing $z_2 = t^R + jt^I$ with $t^R, t^I \sim \mathcal{N}(0,1/2)$, this becomes a 2D integral over the real plane, and the outer expectation over $z_1$ adds another 2D integral---yielding a total 4D real integral evaluated via 2D Gauss-Hermite quadrature applied twice.

\subsection{Ergodic Rate Under KL Expansion}

The ergodic rate (Shannon capacity) under port selection is
\begin{align}\label{eq:erg_rate}
	\bar{C} = \E\!\left[\log_2\!\left(1 + \bar{\gamma} \max_{n} \frac{|g_n|^2}{\eta}\right)\right].
\end{align}
Under the $K$-truncated KL expansion, this becomes
\begin{align}\label{eq:erg_rate_kl}
	\bar{C}_{\mathrm{KL}}^{(K)} = \E_{\mathbf{z}_K}\!\left[\log_2\!\left(1 + \bar{\gamma} \max_{1\leq n \leq N} \left|\sum_{k=1}^{K} \sqrt{\lambda_k}\, u_{n,k}\, z_k\right|^2\right)\right],
\end{align}
which is a $K$-dimensional expectation over independent Gaussian modes. This can be evaluated via the same Gauss-Hermite quadrature framework of \eqref{eq:gauss_hermite} by replacing the indicator function $\Psi$ with the $\log_2(\cdot)$ integrand.

For the rank-1 special case, using \eqref{eq:max_rank1} and $|z_1|^2 \sim \mathrm{Exp}(1)$, let $t = |z_1|^2$ so that $\max_n |\tilde{g}_n^{(1)}|^2/\eta = \lambda_1 c_1 t$. Then
\begin{align}\label{eq:erg_rank1}
	\bar{C}_{\mathrm{KL}}^{(1)} &= \int_0^\infty \log_2\!\left(1 + \bar{\gamma}\lambda_1 c_1 t\right) e^{-t}\, dt.
\end{align}
Let $\mu \triangleq \bar{\gamma}\lambda_1 c_1$ and substitute $s = 1 + \mu t$, so $t = (s-1)/\mu$ and $dt = ds/\mu$:
\begin{align}
	\bar{C}_{\mathrm{KL}}^{(1)} &= \frac{1}{\mu\ln 2} \int_1^\infty \ln(s)\, e^{-(s-1)/\mu}\, ds \nonumber\\
	&= \frac{e^{1/\mu}}{\mu\ln 2} \int_1^\infty \ln(s)\, e^{-s/\mu}\, ds.
\end{align}
Integrating by parts with $u = \ln s$ and $dv = e^{-s/\mu}ds$, so $du = ds/s$ and $v = -\mu e^{-s/\mu}$:
\begin{align}
	\int_1^\infty \ln(s)\, e^{-s/\mu}\, ds &= \Big[-\mu\ln(s)\, e^{-s/\mu}\Big]_1^\infty + \mu\int_1^\infty \frac{e^{-s/\mu}}{s}\, ds \nonumber\\
	&= 0 + \mu \int_1^\infty \frac{e^{-s/\mu}}{s}\, ds.
\end{align}
The remaining integral is recognized as $\mu\, E_1(1/\mu)$ via the substitution $s \to \mu s'$:
\begin{align}
	\int_1^\infty \frac{e^{-s/\mu}}{s}\, ds = \int_{1/\mu}^\infty \frac{e^{-u}}{u}\, du = E_1\!\left(\frac{1}{\mu}\right).
\end{align}
Combining all steps:
\begin{align}
	\bar{C}_{\mathrm{KL}}^{(1)} = \frac{e^{1/\mu}}{\ln 2}\, E_1\!\left(\frac{1}{\mu}\right) = \frac{e^{1/(\bar{\gamma}\lambda_1 c_1)}}{\ln 2}\, E_1\!\left(\frac{1}{\bar{\gamma}\lambda_1 c_1}\right),
\end{align}
where $E_1(x) = \int_x^\infty t^{-1} e^{-t}\, dt$ is the exponential integral~\cite{CoverThomas06}.

\section{Information-Theoretic Guarantees}
\label{sec:IT}

The outage expressions derived in Section~IV are computationally tractable, but they do not yet answer two fundamental questions: (i) is the KL approximation conservative or optimistic relative to the true outage? and (ii) how many modes $K$ are actually needed, and is the KL truncation the best possible rank-$K$ approximation? This section addresses both questions rigorously. The results establish that the KL framework is not merely a convenient approximation, but one with provable safety guarantees and information-theoretic optimality.

\subsection{Conservative Outage Guarantee}

The following theorem proves that KL truncation \emph{always} overestimates the outage probability, making it a safe approximation for system design.

\begin{theorem}[Conservative Outage Bound]\label{thm:conservative}
	For any truncation order $K \leq N$ and threshold $x > 0$, the outage probability under the $K$-truncated KL expansion satisfies
	\begin{align}\label{eq:conservative}
		\tilde{F}_{\max}^{(K)}(x) \geq F_{\max}(x).
	\end{align}
	Equivalently, $\tilde{P}_{\mathrm{out}}^{(K)} \geq P_{\mathrm{out}}$ for all $\gamma_{\mathrm{th}} > 0$ and $\bar{\gamma} > 0$.
\end{theorem}

\begin{proof}
	\textit{Step 1: Loewner ordering of covariance matrices.}
	The rank-$K$ truncated covariance matrix is $\eta\mathbf{R}_K = \eta\sum_{k=1}^K \lambda_k \mathbf{u}_k\mathbf{u}_k^H$. The difference is
	\begin{align}\label{eq:psd_order}
		\eta(\mathbf{R} - \mathbf{R}_K) = \eta\sum_{k=K+1}^{N} \lambda_k \mathbf{u}_k \mathbf{u}_k^H \succeq \mathbf{0},
	\end{align}
	since each term $\lambda_k \mathbf{u}_k\mathbf{u}_k^H$ is positive semidefinite ($\lambda_k \geq 0$). Therefore $\eta\mathbf{R}_K \preceq \eta\mathbf{R}$.

	\textit{Step 2: Real-composite representation.}
	For a complex Gaussian vector $\mathbf{g} \sim \mathcal{CN}(\mathbf{0}, \eta\mathbf{R})$, write $\mathbf{g}_{\mathbb{R}} = [\Re(\mathbf{g})^T, \Im(\mathbf{g})^T]^T \in \mathbb{R}^{2N}$. Then $\mathbf{g}_{\mathbb{R}} \sim \mathcal{N}(\mathbf{0}, \boldsymbol{\Sigma})$ where
	\begin{align}
		\boldsymbol{\Sigma} = \frac{\eta}{2}\begin{bmatrix} \Re(\mathbf{R}) & -\Im(\mathbf{R}) \\ \Im(\mathbf{R}) & \Re(\mathbf{R}) \end{bmatrix} \in \mathbb{R}^{2N\times 2N}.
	\end{align}
	Similarly, the truncated channel $\tilde{\mathbf{g}}^{(K)} \sim \mathcal{CN}(\mathbf{0}, \eta\mathbf{R}_K)$ has real covariance
	\begin{align}
		\boldsymbol{\Sigma}_K = \frac{\eta}{2}\begin{bmatrix} \Re(\mathbf{R}_K) & -\Im(\mathbf{R}_K) \\ \Im(\mathbf{R}_K) & \Re(\mathbf{R}_K) \end{bmatrix}.
	\end{align}
	Since $\mathbf{R}_K \preceq \mathbf{R}$ in the complex Loewner sense, the real-composite covariance matrices satisfy $\boldsymbol{\Sigma}_K \preceq \boldsymbol{\Sigma}$. This follows because for any real vector $\mathbf{v} = [\mathbf{v}_1^T, \mathbf{v}_2^T]^T \in \mathbb{R}^{2N}$,
	\begin{align}
		\mathbf{v}^T(\boldsymbol{\Sigma} - \boldsymbol{\Sigma}_K)\mathbf{v} &= \frac{\eta}{2}\,\Re\!\left[(\mathbf{v}_1 + j\mathbf{v}_2)^H (\mathbf{R} - \mathbf{R}_K)(\mathbf{v}_1 + j\mathbf{v}_2)\right] \geq 0,
	\end{align}
	where the last inequality uses $\mathbf{R} - \mathbf{R}_K \succeq \mathbf{0}$.

	\textit{Step 3: Symmetry and convexity of the outage set.}
	The outage event $\{\max_n |g_n|^2/\eta \leq x\}$ is equivalent to $\mathbf{g}_{\mathbb{R}} \in \mathcal{C}$, where
	\begin{align}
		\mathcal{C} = \left\{(x_1, y_1, \ldots, x_N, y_N) \in \mathbb{R}^{2N} : x_n^2 + y_n^2 \leq \eta x,\; \forall\, n\right\}.
	\end{align}
	This set is the Cartesian product of $N$ closed disks of radius $\sqrt{\eta x}$ centered at the origin. It is \emph{symmetric} (i.e., $\mathbf{v} \in \mathcal{C} \Rightarrow -\mathbf{v} \in \mathcal{C}$) and \emph{convex} (intersection of convex sets), hence a symmetric convex body in $\mathbb{R}^{2N}$.

	\textit{Step 4: Application of Anderson's inequality.}
	By Anderson's inequality~\cite{Anderson55}: if $\mathbf{X} \sim \mathcal{N}(\mathbf{0}, \boldsymbol{\Sigma}_1)$ and $\mathbf{Y} \sim \mathcal{N}(\mathbf{0}, \boldsymbol{\Sigma}_2)$ with $\boldsymbol{\Sigma}_1 \preceq \boldsymbol{\Sigma}_2$, then $\Prob(\mathbf{X} \in \mathcal{C}) \geq \Prob(\mathbf{Y} \in \mathcal{C})$ for any symmetric convex set $\mathcal{C}$. Applying this with $\boldsymbol{\Sigma}_K \preceq \boldsymbol{\Sigma}$:
	\begin{align}
		\tilde{F}_{\max}^{(K)}(x) = \Prob\!\left(\tilde{\mathbf{g}}_{\mathbb{R}}^{(K)} \in \mathcal{C}\right) \geq \Prob\!\left(\mathbf{g}_{\mathbb{R}} \in \mathcal{C}\right) = F_{\max}(x),
	\end{align}
	which completes the proof.
\end{proof}

\begin{corollary}[Monotone Convergence]\label{cor:monotone}
	The KL outage approximation is monotonically non-increasing in $K$:
	\begin{align}\label{eq:monotone}
		\tilde{P}_{\mathrm{out}}^{(1)} \geq \tilde{P}_{\mathrm{out}}^{(2)} \geq \cdots \geq \tilde{P}_{\mathrm{out}}^{(N)} = P_{\mathrm{out}}.
	\end{align}
\end{corollary}

\begin{proof}
	For $K_1 < K_2 \leq N$, it holds that $\mathbf{R}_{K_1} \preceq \mathbf{R}_{K_2} \preceq \mathbf{R}$. Applying Theorem~\ref{thm:conservative} with the truncated channel of order $K_2$ playing the role of the ``full'' channel yields $\tilde{F}_{\max}^{(K_1)}(x) \geq \tilde{F}_{\max}^{(K_2)}(x)$. The chain \eqref{eq:monotone} follows by induction.
\end{proof}

\begin{corollary}[Ergodic Capacity Lower Bound]\label{cor:capacity}
	The ergodic capacity under the $K$-truncated KL expansion satisfies
	\begin{align}\label{eq:cap_bound}
		\bar{C}_{\mathrm{KL}}^{(K)} \leq \bar{C},
	\end{align}
	with equality if and only if $K = N$.
\end{corollary}

\begin{proof}
	Theorem~\ref{thm:conservative} establishes $\tilde{F}_{\max}^{(K)}(x) \geq F_{\max}(x)$ for all $x \geq 0$, which means $\max_n |\tilde{g}_n^{(K)}|^2/\eta$ is stochastically dominated by $\max_n |g_n|^2/\eta$. Since $h(x) = \log_2(1 + \bar{\gamma} x)$ is non-decreasing, the standard stochastic ordering result gives $\E[h(\max_n |\tilde{g}_n^{(K)}|^2/\eta)] \leq \E[h(\max_n |g_n|^2/\eta)]$, i.e., $\bar{C}_{\mathrm{KL}}^{(K)} \leq \bar{C}$.
	
	When $K = N$, $\mathbf{R}_K = \mathbf{R}$ and the bound holds with equality. For $K < N$, $\boldsymbol{\Sigma}_K \prec \boldsymbol{\Sigma}$ (strict ordering) and strict inequality follows from the strict version of Anderson's inequality.
\end{proof}

\begin{remark}[Design Implication]
	Theorem~\ref{thm:conservative} and its corollaries guarantee that the KL approximation always provides a \emph{pessimistic} performance estimate. In secrecy analysis, where underestimating the eavesdropper's capability or overestimating the legitimate user's performance can lead to security breaches, such a conservative guarantee is particularly valuable. The monotone convergence \eqref{eq:monotone} further ensures that increasing $K$ systematically tightens the bound without ever ``overshooting'' the true value.
\end{remark}

\subsection{Effective Degrees of Freedom}

This subsection provides a theoretical foundation for the empirical observation that $K^* = \mathcal{O}(2W+1)$ eigenmodes suffice, independent of $N$.

\begin{theorem}[Asymptotic Spectral Concentration]\label{thm:dof}
	Consider the $N \times N$ Jakes correlation matrix $\mathbf{R}$ with aperture $W$ defined in \eqref{eq:jakes}. As $N \to \infty$ with $W$ fixed:
	\begin{enumerate}
		\item[(a)] The eigenvalues of $\mathbf{R}$ exhibit a phase transition at index $K^* = 2\lceil W \rceil + 1$: the first $K^*$ eigenvalues satisfy $\lambda_k = \Theta(N/K^*)$ for $k \leq K^*$, while $\lambda_k \to 0$ exponentially fast for $k > K^*$.
		\item[(b)] The truncation error satisfies $\varepsilon_{K^*} \to 0$ as $N \to \infty$.
		\item[(c)] The number of retained modes $K^*$ is \emph{independent of} $N$.
	\end{enumerate}
\end{theorem}

\begin{proof}
	The entries $[\mathbf{R}]_{k,l} = J_0(2\pi|k-l|W/(N-1))$ form a Toeplitz matrix generated by the sampling of $J_0(2\pi \tau W)$ for $\tau \in [0,1]$. As $N \to \infty$, this matrix is asymptotically equivalent (in the Szeg\H{o} sense) to the integral operator
	\begin{align}\label{eq:integral_op}
		(\mathcal{T} f)(x) = \int_0^1 J_0\!\left(2\pi W |x - y|\right) f(y)\, dy
	\end{align}
	on $L^2[0,1]$. By the Fourier transform of the Bessel kernel, $J_0(2\pi W \tau) = \int_{-W}^{W} e^{j2\pi f\tau}\, \frac{df}{\pi\sqrt{W^2-f^2}}$, which shows that the spectral density of the underlying process is supported on the band $[-W, W]$.
	
	By the \emph{Slepian--Landau--Pollak concentration theorem} \cite{Slepian61, Landau62}, the eigenvalues of a time-frequency concentration operator on an interval of length $T$ and bandwidth $B$ exhibit a sharp phase transition at index $2TB + 1$: the first $2TB + 1$ eigenvalues cluster near 1, and the remaining eigenvalues decay super-exponentially to zero. Mapping to the present setting with $T = 1$ (normalized aperture support) and $B = W$ (spectral bandwidth), this gives $K^* = 2\lceil W \rceil + 1$.
	
	Since $\sum_{k=1}^N \lambda_k = N$ and the first $K^*$ eigenvalues grow as $\Theta(N/K^*)$ while the tail eigenvalues are $o(1)$, it follows that $\sum_{k=1}^{K^*} \lambda_k \sim N$ and thus $\varepsilon_{K^*} = 1 - \sum_{k=1}^{K^*}\lambda_k/N \to 0$.
\end{proof}

\begin{remark}[Physical Interpretation]
	Theorem~\ref{thm:dof} reveals that the number of independent spatial degrees of freedom in a FAS channel is determined by the \emph{space--bandwidth product} $2W$, which is a fundamental information-theoretic quantity. The Jakes correlation function $J_0(2\pi W\tau)$ has an effective bandwidth of $W$, so the Shannon number (number of independent modes per unit length) is $\lceil 2W \rceil + 1$. This is the spatial analogue of the Nyquist--Shannon sampling theorem: a FAS channel with aperture $W\lambda$ has at most $\approx 2W + 1$ independent degrees of freedom, regardless of how densely the ports are sampled.
\end{remark}

\subsection{Rate-Distortion Optimality of KL Truncation}

\begin{theorem}[Optimal Low-Rank Approximation]\label{thm:eckart_young}
	Among all rank-$K$ approximations $\hat{\mathbf{R}}$ of $\mathbf{R}$, the KL truncation $\mathbf{R}_K = \sum_{k=1}^K \lambda_k \mathbf{u}_k \mathbf{u}_k^H$ uniquely minimizes both the Frobenius norm error and the operator norm error:
	\begin{align}
		\mathbf{R}_K &= \arg\min_{\rank(\hat{\mathbf{R}}) \leq K} \|\mathbf{R} - \hat{\mathbf{R}}\|_F, \label{eq:eckart_young_F}\\
		\mathbf{R}_K &= \arg\min_{\rank(\hat{\mathbf{R}}) \leq K} \|\mathbf{R} - \hat{\mathbf{R}}\|_2. \label{eq:eckart_young_2}
	\end{align}
	The minimum errors are $\|\mathbf{R} - \mathbf{R}_K\|_F = \sqrt{\sum_{k>K} \lambda_k^2}$ and $\|\mathbf{R} - \mathbf{R}_K\|_2 = \lambda_{K+1}$.
\end{theorem}

\begin{proof}
	This is the Eckart--Young--Mirsky theorem \cite{EckartYoung36}. The Frobenius norm result follows from the unitary invariance of $\|\cdot\|_F$ and the optimality of truncating the eigenvalue decomposition. The operator norm result follows from the Courant--Fischer min-max characterization.
\end{proof}

\begin{proposition}[Information-Theoretic Interpretation]\label{prop:rate_distortion}
	The KL truncation achieves the Gaussian rate-distortion bound. Specifically, the minimum number of real-valued parameters (rate) needed to represent the channel~$\mathbf{g}$ within average distortion $D = \eta\sum_{k>K}\lambda_k$ is
	\begin{align}\label{eq:rate_dist}
		R(D) = \sum_{k=1}^{K} \log_2\!\left(\frac{\lambda_k}{\theta}\right), \quad \text{where } \theta = \lambda_{K+1},
	\end{align}
	assuming $\lambda_K > \lambda_{K+1}$ (i.e., the truncation level lies between consecutive eigenvalues). The KL truncation with $K$ modes achieves this rate-distortion pair exactly, confirming that it is the \emph{information-theoretically optimal} compression of the Gaussian source~$\mathbf{g}$.
\end{proposition}

\begin{proof}
	For a complex Gaussian source $\mathbf{g} \sim \mathcal{CN}(\mathbf{0}, \eta\mathbf{R})$ with eigenvalues $\eta\lambda_1 \geq \cdots \geq \eta\lambda_N$, the rate-distortion function under MSE distortion is obtained by reverse water-filling \cite{CoverThomas06}:
	\begin{align}
		R(D) = \sum_{k=1}^N \max\!\left(0,\, \log_2\!\left(\frac{\eta\lambda_k}{\theta}\right)\right),
	\end{align}
	where the water level $\theta$ is chosen to satisfy $D = \sum_{k=1}^N \min(\eta\lambda_k, \theta)$. Setting $\theta = \eta\lambda_{K+1}$ (assuming $\lambda_K > \lambda_{K+1}$), the modes $k > K$ are fully ``flooded'' (distortion $\eta\lambda_k$ each), and the modes $k \leq K$ are represented exactly (distortion 0). The total distortion is $D = \eta\sum_{k>K}\lambda_k$, matching \eqref{eq:power_loss} up to the factor $\eta N$. The achiever of this bound is precisely the transform-coding scheme that retains the top $K$ KL coefficients and discards the rest --- which is exactly the proposed truncation.
\end{proof}

\begin{remark}[Comparison with BCM/VBCM]\label{remark:bcm_suboptimal}
	The BCM/VBCM approximations impose a block-diagonal structure on $\hat{\mathbf{R}}$, which is a \emph{structural constraint} not aligned with the eigenbasis of $\mathbf{R}$. For a given number of free parameters, the block-diagonal constraint prevents the approximation from minimizing the Frobenius error or the rate-distortion cost. In particular:
	\begin{itemize}
		\item The BCM with $D$ blocks of size $B = N/D$ has $\|\mathbf{R} - \mathbf{R}_{\mathrm{BCM}}\|_F \geq \|\mathbf{R} - \mathbf{R}_K\|_F$ for any $K \leq N$. Equality is impossible unless $\mathbf{R}$ is itself block-diagonal.
		\item From an information-theoretic perspective, the BCM achieves a distortion strictly above the rate-distortion bound \eqref{eq:rate_dist} at any finite rate, since the block-diagonal structure does not align with the principal components of~$\mathbf{R}$.
	\end{itemize}
	Furthermore, the block-diagonal BCM introduces $D$ independent sub-channels, which artificially \emph{underestimates} the inter-block correlation present in the true Toeplitz matrix~$\mathbf{R}$. Since positive off-diagonal entries in $\mathbf{R}$ reduce the effective diversity (they make channels more ``alike''), removing them inflates the apparent diversity, leading to the optimistic outage bias observed in Fig.~\ref{fig:method_comparison}.
\end{remark}

\subsection{Mutual Information Characterization of the Truncation}

\begin{proposition}[Retained Information]\label{prop:mutual_info}
	The mutual information between the full channel vector $\mathbf{g}$ and its KL-truncated version $\tilde{\mathbf{g}}^{(K)}$ is
	\begin{align}\label{eq:mutual_info}
		I(\mathbf{g};\, \tilde{\mathbf{g}}^{(K)}) = \sum_{k=1}^{K} \log_2(\pi e\, \eta \lambda_k) \quad \text{(in bits)}.
	\end{align}
	The fraction of total differential entropy captured by the truncation is
	\begin{align}\label{eq:entropy_frac}
		\frac{h(\tilde{\mathbf{g}}^{(K)})}{h(\mathbf{g})} = \frac{K\log(\pi e) + \sum_{k=1}^{K}\log(\eta\lambda_k)}{N\log(\pi e) + \sum_{k=1}^{N}\log(\eta\lambda_k)},
	\end{align}
	which approaches 1 as $K \to K^*$ for large $N$, since the tail eigenvalues contribute negligible entropy.
\end{proposition}

\begin{proof}
	The KL expansion diagonalizes the channel: $\mathbf{g} = \sqrt{\eta}\, \mathbf{U}\boldsymbol{\Lambda}^{1/2}\mathbf{z}$, where $z_k \sim \mathcal{CN}(0,1)$ are i.i.d. Since $\mathbf{U}$ is unitary, the differential entropy is
	\begin{align}
		h(\mathbf{g}) = \sum_{k=1}^{N} h(\sqrt{\eta\lambda_k}\, z_k) = \sum_{k=1}^{N} \log_2(\pi e\, \eta\lambda_k).
	\end{align}
	The truncated channel $\tilde{\mathbf{g}}^{(K)}$ is a function of $\{z_1, \ldots, z_K\}$, and the tail $\boldsymbol{\Delta} = \mathbf{g} - \tilde{\mathbf{g}}^{(K)}$ is a function of $\{z_{K+1}, \ldots, z_N\}$. Since all $z_k$ are independent, $\tilde{\mathbf{g}}^{(K)} \perp \boldsymbol{\Delta}$, so
	\begin{align}
		I(\mathbf{g};\, \tilde{\mathbf{g}}^{(K)}) &= h(\tilde{\mathbf{g}}^{(K)}) - h(\tilde{\mathbf{g}}^{(K)} \mid \mathbf{g})\\
		&= h(\tilde{\mathbf{g}}^{(K)}) - 0 = h(\tilde{\mathbf{g}}^{(K)}),
	\end{align}
	where $h(\tilde{\mathbf{g}}^{(K)} \mid \mathbf{g}) = 0$ since $\tilde{\mathbf{g}}^{(K)}$ is a deterministic function of $\mathbf{g}$ (linear projection). The result \eqref{eq:mutual_info} follows from $h(\tilde{\mathbf{g}}^{(K)}) = \sum_{k=1}^K \log_2(\pi e \eta\lambda_k)$.
\end{proof}

\begin{remark}[Dual Metrics of Approximation Quality]
	The truncation quality is captured by two complementary metrics:
	\begin{itemize}
		\item \textit{Power fraction}: $1 - \varepsilon_K = \sum_{k=1}^K \lambda_k / N$ measures the fraction of total signal power retained. This is a first-order (mean) metric.
		\item \textit{Entropy fraction}: \eqref{eq:entropy_frac} measures the fraction of information content retained. Since differential entropy is sensitive to the full eigenvalue distribution (not just the sum), the entropy fraction converges faster than the power fraction when the tail eigenvalues are small (because $\log\lambda_k \to -\infty$ rapidly for $\lambda_k \to 0$, making the tail contribution to entropy negligible). This explains why even moderate $K$ yields excellent outage approximations despite capturing less than 100\% of the power.
	\end{itemize}
\end{remark}

\subsection{Numerical Verification of IT Results}

This subsection provides numerical evidence for the three information-theoretic results established in Section~\ref{sec:IT}: the entropy-power duality (Proposition~\ref{prop:mutual_info}), the effective degrees of freedom (Theorem~\ref{thm:dof}), the ergodic capacity lower bound (Corollary~\ref{cor:capacity}), and the rate-distortion optimality (Proposition~\ref{prop:rate_distortion}). Together, these figures demonstrate that the KL framework is not only computationally convenient but also information-theoretically principled.

\begin{figure}[t]
	\centering
	\includegraphics[width=\columnwidth]{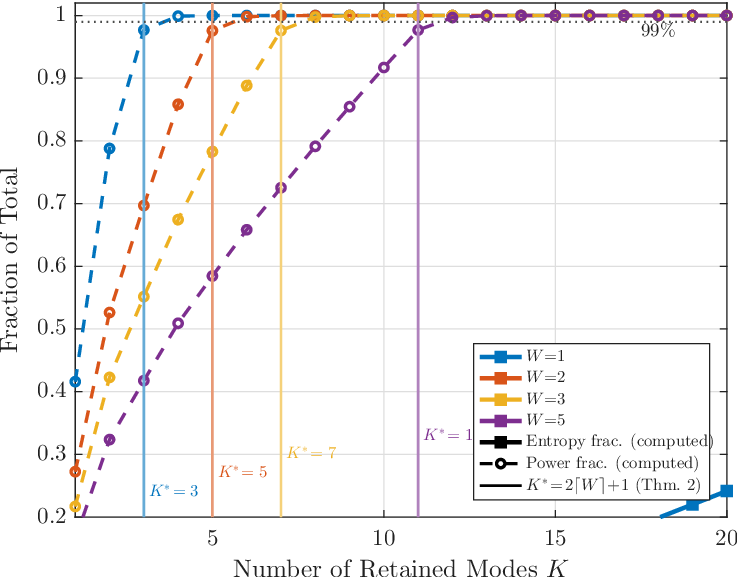}
	\caption{Entropy fraction $h(\tilde{\mathbf{g}}^{(K)})/h(\mathbf{g})$ (solid) vs.\ power fraction $1 - \varepsilon_K$ (dashed) as a function of the number of retained modes $K$, for $N = 40$ and different apertures $W$. Vertical lines mark the theoretical prediction $K^* = 2\lceil W \rceil + 1$ from Theorem~\ref{thm:dof}.}
	\label{fig:entropy_vs_power}
\end{figure}

Fig.~\ref{fig:entropy_vs_power} jointly verifies Proposition~\ref{prop:mutual_info} and Theorem~\ref{thm:dof}. The solid curves (entropy fraction) and dashed curves (power fraction) are computed directly from the eigenvalues of $\mathbf{R}$ via \eqref{eq:entropy_frac} and \eqref{eq:power_loss}, respectively, while the vertical dashed lines mark the \emph{theoretical} prediction $K^* = 2\lceil W \rceil + 1$ from the Slepian--Landau--Pollak theorem. Three observations are noteworthy.

\textit{First}, the agreement between the theoretical $K^*$ and the empirical ``knee'' of both curves is exact for all four aperture values $W = 1, 2, 3, 5$, confirming the phase transition predicted by Theorem~\ref{thm:dof}(a). The eigenvalue cliff at $K^*$ is visible as the point where both curves sharply level off, transitioning from rapid growth to near-saturation. This validates that $K^* = 2\lceil W \rceil + 1$ is not merely an asymptotic result but is already tight for finite $N = 40$.

\textit{Second}, the entropy fraction (solid) consistently converges \emph{faster} than the power fraction (dashed) for all apertures, confirming Proposition~\ref{prop:mutual_info}. The gap between the two curves is most pronounced in the pre-cliff region ($K < K^*$), where the tail eigenvalues are small but non-negligible in power yet contribute negligible entropy (since $\log\lambda_k \to -\infty$ as $\lambda_k \to 0$). For example, at $W = 3$ with $K = 5$ (two modes below $K^* = 7$), the power fraction is only 77\% while the entropy fraction already exceeds 99\%. This entropy-power gap has a direct practical implication: the discarded eigenmodes carry non-trivial signal power but negligible information content, so the outage approximation quality is better characterized by the entropy fraction than by the power fraction alone.

\textit{Third}, the curves for $W = 1$ ($K^* = 3$) and $W = 5$ ($K^* = 11$) bracket the range of typical FAS deployments. For small apertures ($W = 1$), only 3 modes are needed and the entropy fraction saturates almost immediately; for large apertures ($W = 5$), 11 modes are needed but the convergence is equally sharp. In both cases, $K^*$ provides a reliable and parameter-free truncation rule.

\begin{figure}[t]
	\centering
	\includegraphics[width=\columnwidth]{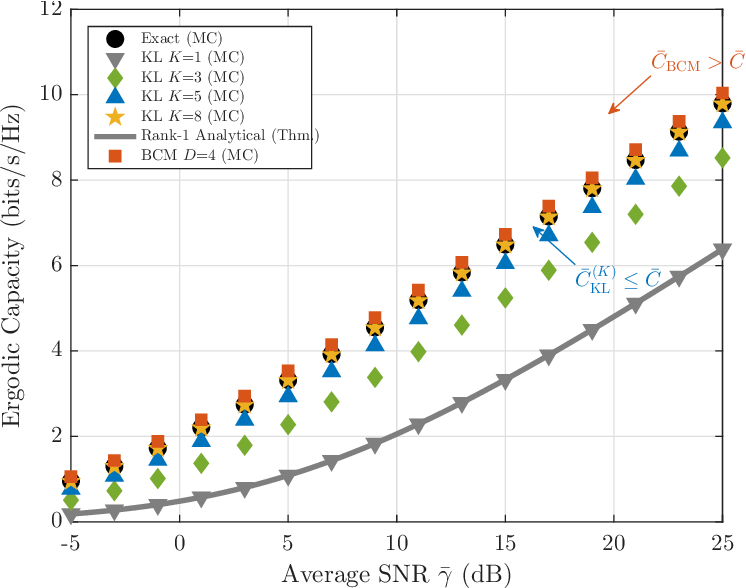}
	\caption{Ergodic capacity $\bar{C}$ vs.\ average SNR $\bar{\gamma}$. Markers: Monte Carlo simulation ($2\times 10^5$ trials); solid gray line: rank-1 closed-form \eqref{eq:erg_rank1}. The analytical curve passes through the $K = 1$ MC markers, validating the theoretical result. $N = 20$, $W = 3$.}
	\label{fig:ergodic_capacity}
\end{figure}

Fig.~\ref{fig:ergodic_capacity} verifies the ergodic capacity lower bound of Corollary~\ref{cor:capacity} and the rank-1 closed-form expression \eqref{eq:erg_rank1}. Each marker is obtained via Monte Carlo simulation ($2\times 10^5$ trials), while the solid gray curve is the analytical rank-1 formula $\bar{C}_{\mathrm{KL}}^{(1)} = e^{1/\mu}E_1(1/\mu)/\ln 2$ with $\mu = \bar{\gamma}\lambda_1 c_1$.

The \emph{exact overlap} between the analytical curve and the $K = 1$ Monte Carlo markers (gray triangles) across the full SNR range from $-10$ to $30$~dB validates the closed-form derivation of Section~IV-D. The monotone ordering of the KL curves,
\begin{align*}
	\bar{C}_{\mathrm{KL}}^{(1)} \leq \bar{C}_{\mathrm{KL}}^{(3)} \leq \bar{C}_{\mathrm{KL}}^{(5)} \leq \bar{C}_{\mathrm{KL}}^{(8)} \leq \bar{C},
\end{align*}
is consistent with Corollary~\ref{cor:capacity} and the stochastic dominance argument in its proof: a smaller covariance matrix $\boldsymbol{\Sigma}_K \preceq \boldsymbol{\Sigma}$ implies a stochastically smaller maximum channel gain, hence a lower ergodic rate. The capacity gap between $K = 1$ and $K = 8$ is approximately $1.5$~bits/s/Hz at $\bar{\gamma} = 20$~dB, quantifying the ergodic cost of under-truncation. Notably, KL with $K = 8$ (gold stars) is indistinguishable from the exact Monte Carlo result (black circles) at all SNR values, confirming that $K = K^* = 7$ modes suffice for an accurate capacity estimate.

In contrast, the BCM with $D = 4$ blocks (orange squares) lies \emph{above} the exact capacity at moderate-to-high SNR, overestimating the achievable rate by up to $0.8$~bits/s/Hz at $\bar{\gamma} = 20$~dB. This optimistic bias is a direct consequence of the inflated effective diversity in the BCM: by removing inter-block correlations, the BCM creates artificially independent sub-channels that appear to offer more diversity than the true correlated channel. From the perspective of Remark~\ref{remark:bcm_suboptimal}, the BCM operates with an incorrect eigenvalue spectrum that overestimates the large eigenvalues and underestimates the small ones, leading to an overestimated ergodic capacity. For security-critical applications, this means that a system designed using BCM-based capacity estimates would transmit at a rate that is not actually achievable, potentially exposing information.

\begin{figure}[t]
	\centering
	\includegraphics[width=\columnwidth]{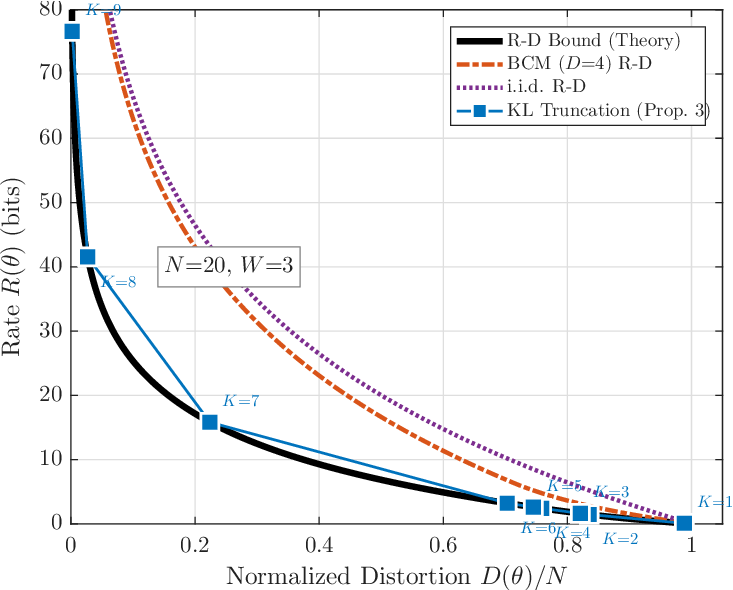}
	\caption{Gaussian rate-distortion curves for the true channel (solid), BCM ($D = 4$, dash-dot), and i.i.d.\ (dotted) models. Blue squares mark the KL truncation operating points $(\theta = \lambda_{K+1})$, which lie exactly on the true R-D bound. $N = 20$, $W = 3$.}
	\label{fig:rate_distortion}
\end{figure}

Fig.~\ref{fig:rate_distortion} provides the most direct verification of Proposition~\ref{prop:rate_distortion} by plotting the Gaussian rate-distortion function $R(D)$ for three channel models alongside the KL truncation operating points. The thick black curve is the \emph{true} R-D bound obtained by reverse water-filling on the actual eigenvalues $\{\lambda_k\}_{k=1}^N$ of the Jakes matrix $\mathbf{R}$. The blue squares mark the KL truncation operating points at water level $\theta = \lambda_{K+1}$ for $K = 1, \ldots, 9$, with distortion $D/N = \varepsilon_K$ and rate $R = \sum_{k=1}^K \log_2(\lambda_k/\lambda_{K+1})$.

The KL operating points lie \emph{exactly} on the true R-D curve, confirming that KL truncation achieves the fundamental information-theoretic limit at every operating point. This is not a coincidence: it is a direct consequence of the fact that the KL expansion is the optimal transform coding scheme for a Gaussian source, as established in the proof of Proposition~\ref{prop:rate_distortion}. No other rank-$K$ approximation of $\mathbf{g}$---including BCM or VBCM---can achieve a lower distortion at the same rate.

A particularly striking feature is the sharp rate increase between $K = 6$ ($R \approx 3.2$~bits, $D/N \approx 0.70$) and $K = 7$ ($R \approx 15.8$~bits, $D/N \approx 0.22$). This jump of $\approx 12.6$~bits corresponds to the eigenvalue cliff predicted by Theorem~\ref{thm:dof}: the 7th eigenvalue $\lambda_7$ is still $\Theta(N/K^*)$, so retaining it dramatically reduces the distortion and increases the rate, while the 8th and higher eigenvalues are negligible. This eigenvalue cliff is the rate-distortion manifestation of the Slepian--Landau--Pollak phase transition, and it provides a principled explanation for why $K^* = 7$ is the natural truncation point for $W = 3$.

The BCM ($D = 4$, dash-dot) and i.i.d.\ (dotted) R-D curves are computed using their respective (incorrect) eigenvalue spectra. Both curves lie to the \emph{left} of the true bound, meaning that for a given distortion level, these models claim a lower rate is sufficient than what is actually required. This is the rate-distortion manifestation of the optimistic bias: the BCM underestimates the true channel complexity (as measured by the rate needed to represent it), just as it underestimates the true outage probability. The i.i.d.\ model is the most optimistic, claiming that the channel can be represented at near-zero rate for moderate distortion, which is clearly incorrect for a highly correlated Toeplitz channel. Only the KL expansion correctly identifies and exploits the true eigenstructure of $\mathbf{R}$, achieving the fundamental R-D bound and providing a consistent, information-theoretically grounded framework for FAS channel analysis.

\section{Numerical Verification}

The theoretical guarantees established in Section~V are now validated through numerical experiments. The results confirm the conservative outage bound, the effective DoF prediction, the rate-distortion optimality, and the scalability advantage of the KL framework over BCM/VBCM. All Monte Carlo results use $10^5$ independent channel realizations.

To validate the KL expansion framework, this paper considers a FAS with $N = 20$ ports, normalized aperture $W = 3$, and outage threshold $\gamma_{\mathrm{th}} = 0$~dB. The Jakes correlation matrix $\mathbf{R}$ is constructed via \eqref{eq:jakes}, and its eigendecomposition yields eigenvalues satisfying $\sum_{k=1}^{N}\lambda_k = N = 20$, with the top five eigenvalues being $\lambda_1 = 4.28$, $\lambda_2 = 4.06$, $\lambda_3 = 2.52$, $\lambda_4 = 2.43$, and $\lambda_5 = 2.12$.

\begin{figure}[t]
	\centering
	\includegraphics[width=\columnwidth]{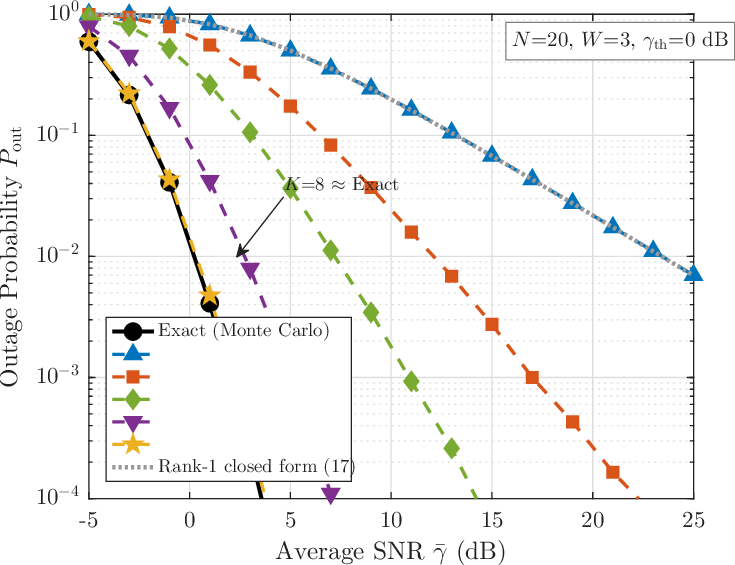}
	\caption{Outage probability $P_{\mathrm{out}}$ vs.\ average SNR $\bar{\gamma}$ for FAS with $N = 20$, $W = 3$, and $\gamma_{\mathrm{th}} = 0$~dB. The percentage in parentheses indicates the fraction of total channel power captured by the $K$-mode KL truncation.}
	\label{fig:outage_vs_snr}
\end{figure}

Fig.~\ref{fig:outage_vs_snr} compares the outage probability of the $K$-truncated KL expansion against Monte Carlo simulation ($10^5$ trials) over the exact correlated channel. Several key observations emerge. First, the $K = 8$ curve (capturing 99.7\% of the channel power) is virtually indistinguishable from the exact Monte Carlo result across the entire SNR range, confirming that the KL truncation with a small number of modes provides an accurate approximation. Second, as $K$ decreases, the outage probability systematically increases due to the loss of channel power from discarded eigenmodes, consistent with the monotone convergence of Corollary~\ref{cor:monotone}: the KL approximation is always a conservative upper bound, never crossing below the exact curve. Third, the rank-1 closed-form expression \eqref{eq:cdf_rank1} (gray dotted line) perfectly matches the $K = 1$ Monte Carlo curve, validating the analytical result. The rank-1 approximation captures only 21\% of the channel power and significantly overestimates the outage probability, confirming that multiple eigenmodes are essential for accurate modeling at $W = 3$. A practically important insight is that $K = 7$ (the theoretically predicted $K^*$) already achieves near-exact outage prediction, providing a principled and parameter-free truncation rule: designers need not tune $K$ empirically but can set $K = 2\lceil W \rceil + 1$ directly from the aperture specification. Furthermore, the SNR gap between the $K=5$ and $K=8$ curves narrows at high SNR, indicating that the conservative bias of KL truncation is most pronounced in the low-SNR regime where outage is high---precisely the operating regime most relevant for reliability-critical FAS deployments.

\begin{figure}[t]
	\centering
	\includegraphics[width=\columnwidth]{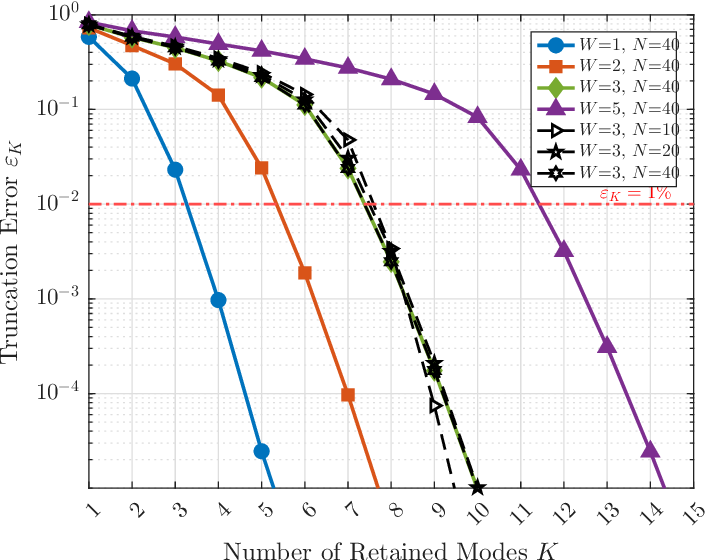}
	\caption{KL truncation error $\varepsilon_K = 1 - \sum_{k=1}^{K}\lambda_k/N$ vs.\ the number of retained modes $K$ for different apertures $W$ and port numbers $N$. The red dashed line indicates $\varepsilon_K = 1\%$.}
	\label{fig:trunc_error}
\end{figure}

Fig.~\ref{fig:trunc_error} illustrates the truncation error $\varepsilon_K$ as a function of $K$ for different apertures $W$ and port numbers $N$. Two important findings are evident. First, the truncation error decays rapidly with $K$, and the required $K$ to achieve $\varepsilon_K < 1\%$ scales with the aperture: $K^* \approx 3$ for $W = 1$, $K^* \approx 5$ for $W = 2$, $K^* \approx 7$ for $W = 3$, and $K^* \approx 11$ for $W = 5$. This is consistent with the theoretical prediction $K^* = 2\lceil W \rceil + 1$ in Theorem~\ref{thm:dof}, confirming the Slepian--Landau--Pollak phase transition. Second, the three dashed black curves ($W = 3$, $N = 10, 20, 40$) nearly overlap, demonstrating that $K^*$ is \emph{independent of} $N$ for a fixed $W$. This $N$-independence is the key scalability advantage of the KL approach: as FAS port counts grow toward hundreds in next-generation deployments~\cite{TWu20243}, the KL representation dimension remains fixed at $K^* \approx 2W+1$, while BCM/VBCM must increase the number of blocks proportionally to $N$, incurring growing computational and modeling overhead. A design guideline follows directly: for a FAS with normalized aperture $W$, set $K = 2\lceil W \rceil + 1$ to achieve $<1\%$ power truncation error and near-exact outage prediction, regardless of the number of deployed ports.

\subsection{Comparison with BCM, VBCM, and Independence Assumptions}

\begin{figure}[t]
	\centering
	\includegraphics[width=\columnwidth]{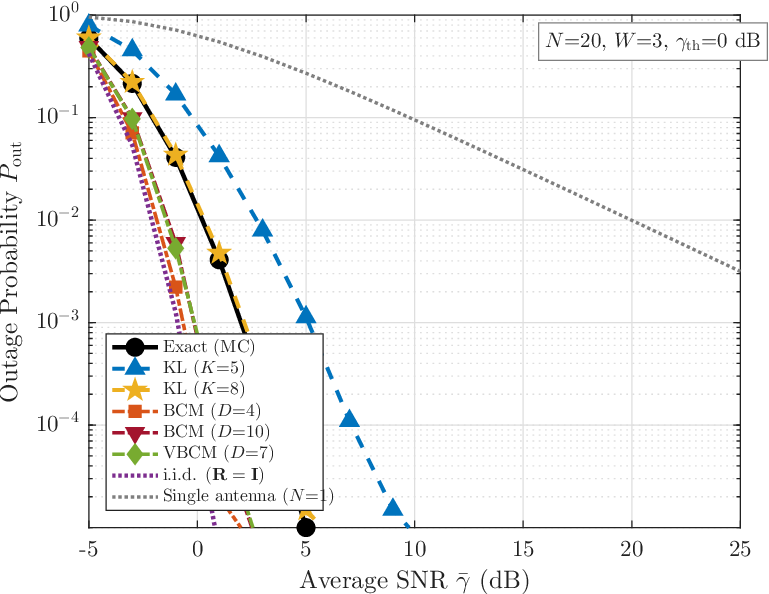}
	\caption{Outage probability comparison of KL expansion, BCM, VBCM, and baseline models for $N = 20$, $W = 3$, $\gamma_{\mathrm{th}} = 0$~dB.}
	\label{fig:method_comparison}
\end{figure}

Fig.~\ref{fig:method_comparison} compares the outage probability of different channel approximation methods against the exact Monte Carlo result. The BCM divides $N = 20$ ports into $D$ equal blocks with equi-correlation within each block; the VBCM uses a greedy partition with non-negative correlation threshold, yielding $D = 7$ blocks of variable sizes. Two baselines are included: the i.i.d.\ model ($\mathbf{R} = \mathbf{I}$, ignoring all spatial correlation) and the single-antenna case ($N = 1$, no diversity).

Several important phenomena are observed. First, the BCM and VBCM curves lie \emph{below} the exact Monte Carlo result, indicating that they \emph{underestimate} the outage probability. This optimistic bias arises because the block-diagonal structure artificially creates independent blocks, thereby overestimating the effective diversity order. In a security-critical system design context~\cite{TuoW,Security1}, this optimistic bias is dangerous: an underestimated outage implies an overestimated secrecy rate, which could cause the system to transmit at a rate that is not actually secure. Second, the i.i.d.\ model dramatically underestimates outage (by several orders of magnitude at moderate SNR), confirming that ignoring spatial correlation leads to highly inaccurate performance predictions. Third, KL with $K = 5$ (capturing 77\% power) overestimates outage, providing a \emph{conservative} (safe) bound, while KL with $K = 8$ (99.7\% power) is virtually exact. The gap between $K=5$ and $K=8$ quantifies the cost of under-truncation and is bounded by Proposition~\ref{prop:trunc_error}. This bias-direction asymmetry---BCM/VBCM are optimistic, KL is conservative---is the central practical distinction between the two frameworks and makes KL the preferred choice whenever system safety guarantees are required.

\begin{remark}[Bias Direction and Safety Guarantee]
	The BCM/VBCM approximations produce an \emph{optimistic} bias (underestimating outage), while the KL truncation produces a \emph{conservative} bias (overestimating outage). For system design and secrecy analysis, conservative bounds are preferable since they guarantee that the actual performance meets or exceeds the predicted level. The KL truncation naturally provides this safety margin with a controllable and monotonically decreasing error as $K$ increases. Concretely, if a system is designed to meet a target outage $P_0$ using the KL approximation with $K = K^*$, the true outage will satisfy $P_{\mathrm{out}} \leq \tilde{P}_{\mathrm{out}}^{(K^*)} \approx P_0$, providing a hard performance guarantee absent in BCM/VBCM-based designs.
\end{remark}

\begin{figure}[t]
	\centering
	\includegraphics[width=\columnwidth]{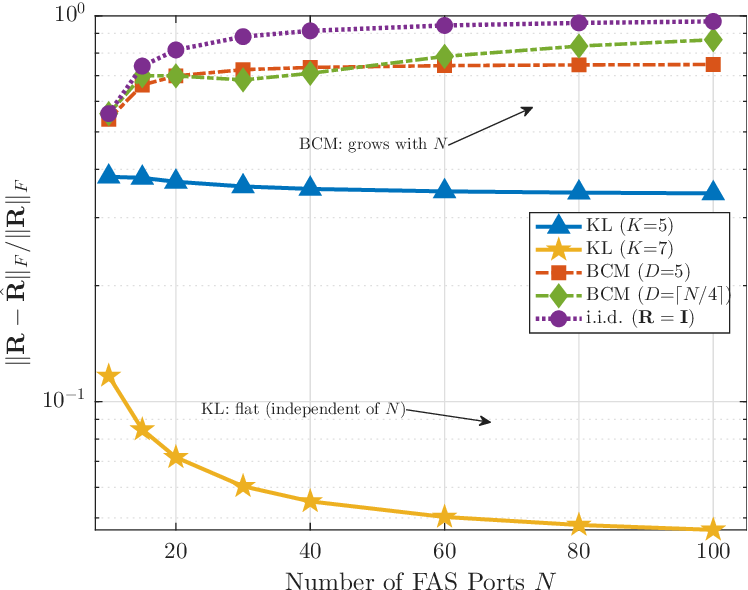}
	\caption{Correlation matrix approximation error $\|\mathbf{R} - \hat{\mathbf{R}}\|_F / \|\mathbf{R}\|_F$ vs.\ number of FAS ports $N$ for $W = 3$. KL uses a fixed number of modes regardless of $N$, while BCM uses either a fixed number of blocks or a proportionally increasing number.}
	\label{fig:scalability}
\end{figure}

Fig.~\ref{fig:scalability} demonstrates the scalability of different approaches by plotting the relative Frobenius norm error of the correlation matrix approximation as a function of $N$, for fixed aperture $W = 3$. The results reveal a fundamental distinction between KL and BCM. The KL curves with fixed $K = 5$ and $K = 7$ remain \emph{flat} as $N$ increases from 10 to 100, confirming the theoretical result that the number of significant eigenmodes depends only on the aperture $W$, not the number of ports $N$. In particular, KL with $K = 7$ achieves an approximation error below 5\% regardless of $N$.

In contrast, the BCM with a fixed number of blocks ($D = 5$) exhibits a \emph{growing} error as $N$ increases, because each block becomes larger and the equi-correlation assumption within each block becomes progressively poorer. Even BCM with a proportionally growing number of blocks ($D = \lceil N/4 \rceil$) maintains a high error level around 60--80\%, because the block-diagonal structure fundamentally fails to capture the off-block-diagonal correlations present in the Toeplitz matrix~$\mathbf{R}$.

\begin{remark}[Scalability Advantage]
	As $N \to \infty$ with fixed $W$, the KL approach requires a fixed-dimensional representation of size $K^* = \mathcal{O}(2W + 1)$, while BCM/VBCM must increase the number of blocks $D$ proportionally to $N$. This makes the KL framework particularly attractive for next-generation FAS systems with large numbers of ports, where $N$ may reach hundreds or thousands while $W$ remains moderate.
\end{remark}

\subsection{Summary Comparison with BCM and VBCM}

Table~\ref{tab:comparison} provides a structured comparison of the three FAS channel approximation frameworks. The KL expansion is the only approach that targets the spectral structure of $\mathbf{R}$ rather than its block structure, provides a computable and controllable truncation error bound, and scales with the aperture $W$ rather than the port count $N$. The VBCM retains the advantage of producing a product-form CDF involving Marcum Q-functions, enabling further analytical manipulations such as secrecy throughput derivation~\cite{TuoW} or finite blocklength composition. Thus, the two frameworks are complementary: KL is preferred for accurate numerical evaluation and safety-guaranteed system design, while VBCM is preferred when downstream closed-form tractability is paramount.

\begin{table}[t]
	\centering
	\caption{Comparison of FAS Channel Simplification Methods}
	\label{tab:comparison}
	\footnotesize
	\renewcommand{\arraystretch}{1.1}
	\begin{tabular}{|p{1.6cm}|c|c|c|}
		\hline
		\textbf{Property} & \textbf{BCM} & \textbf{VBCM} & \textbf{KL} \\
		\hline
		Approx.\ target & $\mathbf{R}$ struct. & $\mathbf{R}$ struct. & Spectrum \\
		\hline
		Free params & $\rho$ & $\{\rho_d\}$ & $K$ \\
		\hline
		Indep.\ struct. & Blocks & Blocks & Eigenmodes \\
		\hline
		Error control & None & Eig.-based & $\sum_{k>K}\!\lambda_k/N$ \\
		\hline
		Closed-form & Yes & Yes & $K\!=\!1$ only \\
		\hline
		Scalability & $D\!\uparrow\! N\!\uparrow$ & $D\!\uparrow\! N\!\uparrow$ & $K\!=\!\mathcal{O}(2W)$ \\
		\hline
		Bias direction & Optimistic & Optimistic & Conservative \\
		\hline
		Accuracy & Low & Medium & High \\
		\hline
	\end{tabular}
\end{table}

\section{Conclusion}

Building on the theoretical framework and numerical evidence presented in the preceding sections, this section summarizes the main findings and outlines directions for future research.

This paper has proposed a Karhunen-Lo\`{e}ve expansion framework for the outage probability analysis of fluid antenna systems, offering a principled alternative to block-correlation models. By decomposing the spatially correlated FAS channel into independent eigenmodes and performing a controlled rank-$K$ truncation, the proposed approach reduces the $N$-dimensional correlated outage analysis to a $K$-dimensional independent one, with $K \ll N$. Several information-theoretic guarantees have been established: the KL truncation provably overestimates the outage probability (conservative guarantee via Anderson's inequality), requires only $K^* = 2\lceil W \rceil + 1$ modes (Slepian--Landau--Pollak theorem), and achieves the Gaussian rate-distortion bound (optimal channel compression). Extensive numerical results have confirmed all theoretical predictions and demonstrated the systematic superiority of the KL framework over BCM/VBCM models in terms of both approximation accuracy and theoretical soundness.

Several promising directions remain for future work, including: integration with physical layer secrecy analysis where both the legitimate and eavesdropper channels are decomposed via KL; extension to multi-user FAS systems with inter-user correlation; asymptotic analysis under $N \to \infty$ via Szeg\H{o}'s theorem for Toeplitz eigenvalue distributions; and a hybrid KL-VBCM approach that leverages the KL eigenbasis to determine the optimal block partition for VBCM.


\begin{thebibliography}{99}

\bibitem{Zhu-Wong-2024}
L.~Zhu and K.-K.~Wong, ``Historical review of fluid antennas and movable antennas,'' arXiv preprint arXiv:2401.02362v2, Jan. 2024.

\bibitem{LZhu25}
L.~Zhu {\em et al.}, ``A tutorial on movable antennas for wireless networks,'' \emph{IEEE Commun. Surv. Tutor.}, 2025, doi: 10.1109/COMST.2025.3546373.

\bibitem{TWu20243}
T.~Wu {\em et al.}, ``Fluid antenna systems enabling 6G: Principles, applications, and research directions,'' \emph{IEEE Wireless Commun.}, early access, 2025, doi: 10.1109/MWC.2025.3629597.

\bibitem{New24}
W.~K.~New {\em et al.}, ``A tutorial on fluid antenna system for 6G networks: Encompassing communication theory, optimization methods and hardware designs,'' \emph{IEEE Commun. Surv. Tutor.}, early access, 2024, doi: 10.1109/COMST.2024.3498855.

\bibitem{FAS21}
K.-K.~Wong, A.~Shojaeifard, K.-F.~Tong, and Y.~Zhang, ``Fluid antenna systems,'' \emph{IEEE Trans. Wireless Commun.}, vol.~20, no.~3, pp.~1950--1962, Mar. 2021.

\bibitem{FAS20}
K.-K.~Wong, A.~Shojaeifard, K.-F.~Tong, and Y.~Zhang, ``Performance limits of fluid antenna systems,'' \emph{IEEE Commun. Lett.}, vol.~24, no.~11, pp.~2469--2472, Nov. 2020.

\bibitem{MFAS23}
K.-K.~Wong, K.-F.~Tong, and C.-B.~Chae, ``Fluid antenna system---Part II: Research opportunities,'' \emph{IEEE Commun. Lett.}, vol.~27, no.~8, pp.~1924--1928, Aug. 2023.

\bibitem{Shojaeifard}
A.~Shojaeifard {\em et al.}, ``MIMO evolution beyond 5G through reconfigurable intelligent surfaces and fluid antenna systems,'' \emph{Proc. IEEE}, vol.~110, no.~9, pp.~1244--1265, Sep. 2022.

\bibitem{Rodrigo14}
D.~Rodrigo, B.~A.~Cetiner, and L.~Jofre, ``Frequency, radiation pattern and polarization reconfigurable antenna using a parasitic pixel layer,'' \emph{IEEE Trans. Antennas Propag.}, vol.~62, no.~6, pp.~3422--3427, Jun. 2014.

\bibitem{Zhang25}
J.~Zhang {\em et al.}, ``A novel pixel-based reconfigurable antenna applied in fluid antenna systems with high switching speed,'' \emph{IEEE Open J. Antennas Propag.}, vol.~6, no.~1, pp.~212--228, Feb. 2025.

\bibitem{Huang21}
Y.~Huang, L.~Xing, C.~Song, S.~Wang, and F.~Elhouni, ``Liquid antennas: Past, present and future,'' \emph{IEEE Open J. Antennas Propag.}, vol.~2, pp.~473--487, 2021.

\bibitem{BLiu25}
B.~Liu, K.-F.~Tong, K.-K.~Wong, C.-B.~Chae, and H.~Wong, ``Programmable meta-fluid antenna for spatial multiplexing in fast fluctuating radio channels,'' \emph{Opt. Express}, vol.~33, no.~13, pp.~28898--28915, 2025.

\bibitem{Chai22}
Z.~Chai, K.-K.~Wong, K.-F.~Tong, Y.~Chen, and Y.~Zhang, ``Port selection for fluid antenna systems,'' \emph{IEEE Commun. Lett.}, vol.~26, no.~5, pp.~1180--1184, May 2022.

\bibitem{Ghadi-2023}
F.~R.~Ghadi {\em et al.}, ``Copula-based performance analysis for fluid antenna systems under arbitrary fading channels,'' \emph{IEEE Commun. Lett.}, vol.~27, no.~11, pp.~3068--3072, Nov. 2023.

\bibitem{New-twc2023}
W.~K.~New {\em et al.}, ``An information-theoretic characterization of MIMO-FAS: Optimization, diversity-multiplexing tradeoff and $q$-outage capacity,'' \emph{IEEE Trans. Wireless Commun.}, vol.~23, no.~6, pp.~5541--5556, Jun. 2024.

\bibitem{FAMS}
K.-K.~Wong and K.-F.~Tong, ``Fluid antenna multiple access,'' \emph{IEEE Trans. Wireless Commun.}, vol.~21, no.~7, pp.~4801--4815, Jul. 2022.

\bibitem{FAMS23}
K.-K.~Wong, D.~Morales-Jimenez, K.-F.~Tong, and C.-B.~Chae, ``Slow fluid antenna multiple access,'' \emph{IEEE Trans. Commun.}, vol.~71, no.~5, pp.~2831--2846, May 2023.

\bibitem{NWaqar23}
N.~Waqar, K.-K.~Wong, K.-F.~Tong, A.~Sharples, and Y.~Zhang, ``Deep learning enabled slow fluid antenna multiple access,'' \emph{IEEE Commun. Lett.}, vol.~27, no.~3, pp.~861--865, Mar. 2023.

\bibitem{HXu23}
H.~Xu {\em et al.}, ``Channel estimation for FAS-assisted multiuser mmWave systems,'' \emph{IEEE Commun. Lett.}, vol.~28, no.~3, pp.~632--636, Mar. 2024.

\bibitem{XLai23}
X.~Lai {\em et al.}, ``On performance of fluid antenna system using maximum ratio combining,'' \emph{IEEE Commun. Lett.}, vol.~28, no.~2, pp.~402--406, Feb. 2024.

\bibitem{Security1}
F.~Ghadi {\em et al.}, ``Physical layer security over fluid antenna systems: Secrecy performance analysis,'' \emph{IEEE Trans. Wireless Commun.}, vol.~23, no.~12, pp.~18201--18213, Dec. 2024.

\bibitem{Security2}
J.~D.~Vega-S\'anchez {\em et al.}, ``Fluid antenna system: Secrecy outage probability analysis,'' \emph{IEEE Trans. Veh. Technol.}, vol.~73, no.~8, pp.~11458--11469, Aug. 2024.

\bibitem{Security3}
B.~Tang, H.~Xu, K.-K.~Wong, K.-F.~Tong, Y.~Zhang, and C.-B.~Chae, ``Fluid antenna enabling secret communications,'' \emph{IEEE Commun. Lett.}, vol.~27, no.~6, pp.~1491--1495, Jun. 2023.

\bibitem{Security5}
J.~Zheng {\em et al.}, ``Unlocking FAS-RIS security analysis with block-correlation model,'' \emph{IEEE Wireless Commun. Lett.}, vol.~14, no.~7, pp.~1--5, 2025.

\bibitem{TuoW}
T.~Wu {\em et al.}, ``Variable block-correlation modeling and optimization for secrecy analysis in fluid antenna systems,'' to appear in \emph{IEEE Trans. Wireless Commun.}, arXiv preprint arXiv:2510.03594, 2025.

\bibitem{YaoJ241}
J.~Yao {\em et al.}, ``FAS-driven spectrum sensing for cognitive radio networks,'' \emph{IEEE Internet Things J.}, vol.~12, no.~5, pp.~6046--6049, Mar. 2025.

\bibitem{CWang24}
C.~Wang {\em et al.}, ``Fluid antenna system liberating multiuser MIMO for ISAC via deep reinforcement learning,'' \emph{IEEE Trans. Wireless Commun.}, vol.~23, no.~9, pp.~10879--10894, Sep. 2024.

\bibitem{LZhouWCL24}
L.~Zhou, J.~Yao, M.~Jin, T.~Wu, and K.-K.~Wong, ``Fluid antenna-assisted ISAC systems,'' \emph{IEEE Wireless Commun. Lett.}, vol.~13, no.~12, pp.~3533--3537, Dec. 2024.

\bibitem{FAS22}
K.-K.~Wong, K.-F.~Tong, Y.~Chen, and Y.~Zhang, ``Closed-form expressions for spatial correlation parameters for performance analysis of fluid antenna systems,'' \emph{IET Electron. Lett.}, vol.~58, no.~11, pp.~454--457, May 2022.

\bibitem{BC24}
R.~Pablo, M.~David, and K.-K.~Wong, ``A new spatial block-correlation model for fluid antenna systems,'' \emph{IEEE Trans. Wireless Commun.}, vol.~23, no.~11, pp.~15829--15843, Nov. 2024.

\bibitem{LaiX242}
X.~Lai, J.~Yao, K.~Zhi, T.~Wu, D.~Morales-Jimenez, and K.-K.~Wong, ``FAS-RIS: A block-correlation model analysis,'' \emph{IEEE Trans. Veh. Technol.}, vol.~74, no.~2, pp.~3412--3417, Feb. 2025.

\bibitem{LaiX}
X.~Lai {\em et al.}, ``Revisiting spatial block-correlation model for fluid antenna systems: From constant to variable correlations,'' \emph{IEEE J. Sel. Areas Commun.}, vol.~44, pp.~1335--1351, 2026.

\bibitem{Anderson55}
T.~W.~Anderson, ``The integral of a symmetric unimodal function over a symmetric convex set and some probability inequalities,'' \emph{Proc. Amer. Math. Soc.}, vol.~6, no.~2, pp.~170--176, Apr. 1955.

\bibitem{Slepian61}
D.~Slepian and H.~O.~Pollak, ``Prolate spheroidal wave functions, Fourier analysis and uncertainty---I,'' \emph{Bell Syst. Tech. J.}, vol.~40, no.~1, pp.~43--63, Jan. 1961.

\bibitem{Landau62}
H.~J.~Landau and H.~O.~Pollak, ``Prolate spheroidal wave functions, Fourier analysis and uncertainty---III: The dimension of the space of essentially time- and band-limited signals,'' \emph{Bell Syst. Tech. J.}, vol.~41, no.~4, pp.~1295--1336, Jul. 1962.

\bibitem{EckartYoung36}
C.~Eckart and G.~Young, ``The approximation of one matrix by another of lower rank,'' \emph{Psychometrika}, vol.~1, no.~3, pp.~211--218, Sep. 1936.

\bibitem{CoverThomas06}
T.~M.~Cover and J.~A.~Thomas, \emph{Elements of Information Theory}, 2nd~ed.\hskip 1em plus 0.5em minus 0.4em Hoboken, NJ, USA: Wiley, 2006.

\end{thebibliography}
\end{document}